\documentclass[12pt]{article}
\usepackage{graphicx}
\usepackage{amsfonts,amsmath,amsthm,amstext,amssymb,url}
\usepackage{color}
\usepackage[dvipsnames]{xcolor}
\usepackage{authblk}
\usepackage{hyperref}
\usepackage[applemac]{inputenc} 
\usepackage{caption}
\usepackage{float}

\def\R{\mathbb{R}}
 \newtheorem{theorem}{\sffamily Theorem}
\newtheorem{lemma}[theorem]{\sffamily Lemma}

\newtheorem{proposition}[theorem]{\sffamily Proposition}
\newtheorem{definition}[theorem]{\sffamily Definition}
\newtheorem{remark}[theorem]{\sffamily Remark}

\title{Relative position of a parabola or a hyperbola and an ellipse without computing intersection points}

\date{}

\begin{document}

\author[1]{Jorge Caravantes}
\author[2]{Gema M. Diaz--Toca}
\author[3]{Mario Fioravanti}
\author[4]{Laureano Gonzalez--Vega\thanks{The authors have been partially supported by the grant PID2020-113192GB-I00/AEI/10.13039/501100011033 (Mathematical Visualization: Foundations, Algorithms and Applications) from the Spanish State Research Agency (Ministerio de Ciencia e Innovaci\'on).}}
\affil[1]{Universidad de Alcal\'a (Spain)\\ \texttt{jorge.caravantes@uah.es}}
\affil[2]{Universidad de Murcia (Spain)\\ \texttt{gemadiaz@um.es}}
\affil[3]{Universidad de Cantabria (Spain)\\ \texttt{mario.fioravanti@unican.es}}
\affil[4]{CUNEF Universidad (Spain)\\ \texttt{laureano.gonzalez@cunef.edu}}

\maketitle

\begin{abstract}
Efficient methods to determine the relative position of two conics are of great interest for applications in robotics, computer animation, CAGD, computational physics, and other areas. We present a method to obtain the relative position of a parabola or a hyperbola, and a coplanar ellipse, directly from the coefficients of their implicit equations, even if they are not given in canonical form, and avoiding the computation of the corresponding intersection points (and their characteristics). 
\end{abstract}

\section*{Introduction}

The problem of detecting the collisions or overlap of two conics in the plane is of interest to robotics, CAD/CAM, computational physics, computer animation, etc., where conics are often used for modelling (or enclosing) the shape of the objects under consideration. 

We present an efficient method to determine the relative position of either a parabola $\mathcal{N}$ or a hyperbola $\mathcal{NH}$, and an ellipse $\mathcal{M}$, in terms of the sign of several polynomial expressions derived from the coefficients of the implicit equations of $\mathcal{N}$ and $\mathcal{M}$. The main advantage of this approach relies on the fact that we do not require to compute the  intersection points of $\mathcal{N}$ (or $\mathcal{NH}$) and $\mathcal{M}$, nor their characteristics, and it is specially useful when $\mathcal{N}$ (or $\mathcal{NH}$) and $\mathcal{M}$ depend on one or several parameters.

Several authors have considered the case of two coplanar ellipses. In \cite{Etayo-et-al:2006} the properties of the pencil of two ellipses depending on one parameter are studied using Sturm-Habicht sequences, and ten types of relative positions are characterized. Using pencils of two ellipses analogously, in \cite{AET} a method to determine the relative position of two ellipses is explained obtaining ten different classes, which do not fully coincide with the previous classification.

Some characterization methods can be derived from the relative position of two quadrics, considering a suitable sectioning; in \cite{Wang-Wang-Kim:2001} they determine if two ellipsoids are separated or they touch tangentially, depending on the root configuration of their characteristic equation. Similarly, by sectioning the different cases of relative position of a one sheet hyperboloid and a sphere in \cite{Brozos:2018}, it may be obtained the corresponding algebraic characterization for the relative position of a hyperbola and a circle.

We follow here an approach similar to the one presented in \cite{Min-Chen:2008} and \cite{Liu-Chen:2004}. 
The problem considered here can be presented as a quantifier elimination problem over the reals (see \cite{BPR}), since we are looking for conditions on the coefficients of the equations defining the considered conics in order they produce a prescribed geometric configuration. The bridge between the problem at hand and quantifier elimination is the characteristic equation of the two considered conics, since the properties of its real roots determine each geometric configuration.

The paper is organized as follows. After some preliminaries, Section 2 introduces the necessary criteria to discern the relative position between parabola and ellipse. In Section 3, we discuss and prove the necessary criteria to discern the relative position between ellipse and hyperbola. Finally, Section 4 contains the conclusions. 

\section{Discriminants and Descartes' law of signs}

Given a polynomial $P=a_p T^p+ \ldots + a_0\in \mathbb{C}[T]$, $a_p\neq 0, p>0$, the discriminant of $P$ is defined to be
$$\text{Disc}(P)=\frac{(-1)^{\frac{p(p-1)}{2}}}{a_p}\mathrm{\bf Res}(P,P'),$$
where $\mathrm{\bf Res}(P,P')$ denotes the resultant of $P$ and its derivative.
If $P$ is monic and $x_1,\ldots,x_p$ are the roots of $P$ (repeated according to their multiplicities), it is well known that the discriminant is equal to
$$
\text{Disc}(P)=\prod\limits_{p\geq i > j \geq 1}(x_i-x_j)^2.
$$
Therefore, the use of the discriminant is typically linked to determine when a polynomial has multiple roots but, for low degree polynomials, it is also very useful to determine easily the number of different real roots (see \cite{Janson} for example). The next (very) well known lemma shows that the positivity of the discriminant characterizes the case of having three different real roots (see \cite{Janson}).

\begin{lemma}\label{ellipsesdisc}\hskip 0pt\\
Let $P(T)=a_3T^3+a_2T^2+a_1T+a_0$ be a polynomial in $\R[T]$ with $a_3\neq 0$,and
$${\rm Disc}(P)=  -27a_0^2a_3^2 + 18a_0a_1a_2a_3 - 4a_0a_2^3 - 4a_1^3a_3 + a_1^2a_2^2$$ its discriminant.
$P(T)$ has three different real roots if and only if $\rm{Disc}(P)>0$.\end{lemma}


\subsection{Descartes' law of signs}

For a sequence of real numbers $b_0,b_1,\ldots,b_n$, ${\bf Var}(b_0,b_1,\ldots,b_n)$ will denote the number of sign changes in $b_0,b_1,\ldots,b_n$ after dropping the zeros in the sequence.

\begin{proposition}\hskip 0pt\\
Let $P$ be the polynomial in $\R[x]$,
\[P(x)=\sum_{k=0}^n a_kx^k,\]
where $a_n$ and $a_0$ are nonzero.
The number of positive real roots of $P(x)=0$, counted with multiplicity, is equal to $${\bf Var}(a_n,a_{n-1},\ldots,a_0)-2k$$ for some non--negative integer $k$.
\end{proposition}

\begin{proposition}\label{realr}\hskip 0pt\\
Let $P$ be the polynomial in $\R[x]$, we call $ v$ the number of changes of signs of the list of its coefficients and we call $v'$ the number of changes of signs of the list of the coefficients of the polynomial $P(-x)$. Then, if all the roots of $P$ are real and nonzero, the number $r$ of positive roots of $P$ and the number $r'$ of its negative roots satisfy $$r=v \text{ and } r'=v'\, .$$
\end{proposition}
See Remark 2.42 in \cite{BPR}, and \cite{Mignotte:1991}, p. 199, for a proof.

\subsection{The characteristic equation of a pencil of conics}\label{cheq}
The following definitions will play a crucial role in what follows. Let ${\cal X}=(X ,\, Y  ,\, 1 )$.
\begin{definition}\hskip 0pt\\
Given two conics $\mathcal{A}: {\cal X}{\bf M}{\cal X}^{t}=0$ and $\mathcal{B}:
{\cal X}{\bf N}{\cal X}^{t}=0$, their characteristic equation is defined as
$$ \det(\lambda {\bf N} + {\bf M})=\det({\bf N})\lambda^3+\ldots+\det({\bf M})$$ which is, if $\det({\bf N})\neq 0$, a cubic polynomial in $\lambda$ with real coefficients.
\end{definition}

Suppose now that ${\bf M}$ and ${\bf N}$ are given by:
$${\bf M}=\begin{pmatrix}
   a_{11} & a_{12} & a_{13} \\
   a_{12} & a_{22} & a_{23} \\
   a_{13} & a_{23} & a_{33}
\end{pmatrix},\quad {\bf N}=\begin{pmatrix}
   b_{11} & b_{12} & b_{13} \\
   b_{12} & b_{22} & b_{23} \\
   b_{13} & b_{23} & b_{33}
\end{pmatrix}.$$
Let
$$L_0=\det {\bf N},\:  L_3=\det(\bf M)$$
and
$${\bf M_2}=\begin{pmatrix}
   a_{11} & a_{12}   \\
   a_{12} & a_{22}  \end{pmatrix}, \:
   {\bf N_2}=\begin{pmatrix}
   b_{11} & b_{12}   \\
   b_{12} & b_{22}\end{pmatrix}, \:
   T_1=\det {\bf M_2} , \: T_2=\det {\bf N_2} .
 $$
If
$$
L_1 = L_0 {\rm Trace( {\bf N}^{-1} {\bf M})},\quad
L_2 = L_3 {\rm Trace( {\bf M}^{-1} {\bf N})}, \quad
T = T_1 {\rm Trace( {\bf M_2}^{-1} {\bf N_2})} ,
$$
then, the characteristic polynomial of $\mathcal{A}$ and $\mathcal{B}$, denoted by $F(\lambda)$, turns out to be (see \cite{Min-Chen:2008})
$$
F(\lambda)=\det(\lambda {\bf N}+ {\bf M})= L_0 \lambda^3+ L_1 \lambda^2+  L_2 \lambda +L_3,
$$
and thus, the discriminant is equal to
$$\Delta=-27 L_0^2L_3^2+18L_0L_1L_2L_3+L_1^2L_2^2-4L_1^3L_3-4L_0L_2^3.
$$

\section{Parabola and ellipse}

In a first step, we will assume that the parabola is in canonical form and the ellipse is a circle. Afterwards we will consider a general parabola and any ellipse, find the transformations that convert the parabola into canonical form and the ellipse into a circle, and rephrase the main theorem accordingly.

\subsection{Relative position. Canonical form}\label{sec:canonical}



Suppose that the parabola $\mathcal{N}$ is giving in canonical form. That is,
\begin{equation*}
{\mathcal X}= \left(x \ y \ 1\right),\quad
{\bf N}=\left( \begin {array}{ccc} {a}^{-2}&0&0\\ \noalign{\medskip}0&0&-1
\\ \noalign{\medskip}0&-1&0\end {array} \right), \quad  \mathcal{N}:\,{\mathcal X}\mathbf{N}{\mathcal X}^{t}=\dfrac{x^2}{a^2}-2y=0,
%
\end{equation*}
and $\mathcal{M}$ is a circle defined by ${\mathcal X}{\bf M}{\mathcal X}^{t}=(x-x_c)^2+(y-y_c)^2-\delta^2=0$,
$${\bf M}=\left( \begin {array}{ccc} 1&0&-{\it x_c}\\ \noalign{\medskip}0&1&-{
\it y_c}\\ \noalign{\medskip}-{\it x_c}&-{\it y_c}&-{\delta}^{2}+{{\it x_c}}^{2}+{{\it y_c}}^{2}\end {array} \right).
$$
Hence, the characteristic equation of ${\bf N}$ and ${\bf M}$ will be denoted by
$$ f(\lambda)=\det(\lambda\,{\bf N}+{\bf M})= c_0\lambda^3+c_1\lambda^2+c_2\lambda+c_3,$$
with
$$c_0=\dfrac{-1}{a^2},\, c_1=\dfrac{-\left( a^2+2\,{\it y_c}\right)}{a^2},\, c_2=-{\dfrac{\left(2\,a^2{\it y_c}+{
\delta}^2-{{\it x_c}}^2 \right)}{a^2}},\, c_3=-\delta^2.
$$
Let $\Delta$ be the discriminant of $f(\lambda)$, $\Delta'$ be the discriminant of $f'(\lambda)$, that is, $\Delta'= -12c_0c_2 + 4c_1^2$,
and
$$g(\lambda)=f(\lambda-a^2)=c_0'\lambda^3+c_1'\lambda^2+c_2'\lambda+c_3',$$
with
 \begin{align*}
 c_0'&=\dfrac{-1}{a^2}, \quad c_1'=c_1-3 c_0 a^2  = \dfrac{2(a^2 - y_c)}{a^2},\\
 c_2'&=3c_0a^4-2c_1a^2+c_2  =\dfrac{ (x_c^2 +2a^2y_c -a^4  - \delta^2 )}{a^2},\quad
c_3'= -x_c^2.
  \end{align*}

In this situation, the next proposition characterizes the relative position of $\mathcal{N}$ and $\mathcal{M}$ in terms of the roots of the characteristic equation of $\bf N$ and $\bf M$.  A point of intersection is said an ``inner tangent point" when the parabola and the circle have the same tangent line at this point, and both curves are on the same side of the tangent line. The first eight cases are introduced by Proposition 2.2 in \cite{Liu-Chen:2004}; however, the ninth case is missing and so we thought it was convenient to introduce a proof. 

\begin{proposition}\label{paracirc}\hskip 0pt\\
Consider the parabola $\mathcal{N}$ and the circle $\mathcal{M}$, as above.
\begin{enumerate}
\item  $\mathcal{M}$ and $\mathcal{N}$  are separated iff $f(\lambda) = 0$ has two
distinct positive roots.
\item  $\mathcal{M}$ and $\mathcal{N}$  are externally tangent iff $f(\lambda) = 0$ has
a positive double root.
\item  $\mathcal{M}$ is inside $\mathcal{N}$ iff $f(\lambda) = 0$ has three distinct
negative roots, two of which are not less than  $-a^2$ and one root belongs to $(-\infty,-a^2)$, or three roots are $-a^2,-a^2,-\delta^2/a^2$ when $a^2 > \delta$.
\item  $\mathcal{M}$ and $\mathcal{N}$ have only two intersection points iff $f(\lambda) = 0$ has two imaginary roots.
\item  $\mathcal{M}$ and $\mathcal{N}$ have four points of intersection iff $f(\lambda) = 0$ has three distinct negative roots which are not greater than $-a^2$.
\item $\mathcal{M}$ and $\mathcal{N}$ have two points of intersection and an inner tangent point iff $f(\lambda) = 0$ has a negative double root different from $-a^2$ and the three roots are not greater than $-a^2$, where $a^2 \leq \delta$.

\item$\mathcal{M}$ and $\mathcal{N}$ have only an inner tangent point iff
\begin{itemize}
\item   $f(\lambda) =0$ has a negative double root which is greater than $-a^2$, or,
\item  $-a^2$ is a triple root of $f(\lambda) = 0$.
\end{itemize}

\item $\mathcal{M}$ and $\mathcal{N}$ have two inner tangent points iff the roots of $f(\lambda) = 0$ are $-a^2$,$-a^2$,$-\delta^2/a^2$ where $a^2<\delta$.
 \item  $\mathcal{M}$ and $\mathcal{N}$ have only an intersection point and an inner tangent point iff $f(\lambda) = 0$ has a triple root, less than $-a^2$.
 \end{enumerate}
 \end{proposition}
\begin{proof}
As we said previously, we must prove Point 9 (see \cite[Proposition 2.2]{Liu-Chen:2004} for the rest).

In general, the parabola $\mathcal{N}$ and the circle $\mathcal{M}$ have four points of intersection in $\mathbb{C}^2$, maybe equal depending on the multiplicities. Suppose that three of these possibly non-real  intersection points are given symbolically as $ p_1=\left( x_1  , \dfrac{ x_1^2}{2 a^2} \right) $, $ p_2=\left( x_2  , \dfrac{ x_2^2}{2 a^2} \right) $ and $ p_3=\left( x_3  , \dfrac{ x_3^2}{2 a^2} \right) $. After some computations, another equation for the circle $\mathcal{M}$ is provided by the circle defined by these three points,
$$
\mathcal{M}: x^{2}+y^{2}+d_1\,x+d_2\,y+d_3=0,
$$
with
\begin{eqnarray*}
d_1&=&\frac{\left(x_2 +x_3 \right) \left(x_1 +x_3 \right) \left(x_1 +x_2 \right) }{4 a^{4}},\\
d_2&=&-\frac{ \left(4 a^{4}+x_1^{2}+x_1 x_2 +x_1 x_3 +x_2^{2}+x_3 x_2 +x_3^{2}\right)}{2 a^{2}},\\
d_3&=&-\frac{x_1 x_2 x_3 \left(x_1 +x_2 +x_3 \right)}{4 a^{4}}.
\end{eqnarray*}

Once we have this equation, the characteristic polynomial of $\mathcal{N}$ and $\mathcal{C}$ can also be expressed  in these terms,
$$
f(\lambda)=\dfrac{-1}{a^2}(\lambda-r_1 )(\lambda-r_2 )(\lambda-r_3 ),$$ 
with 
$$
r_1=-a^2 - \dfrac{(x_1 + x_2)^2}{4a^2}, r_2=-a^2 - \dfrac{(x_1 + x_3)^2}{4a^2}, r_3= 
-a^2- \dfrac{(x_2 + x_3)^2}{4a^2};
$$
In addition, as we mentioned before, the intersection of the circle and the parabola does not consist only of three points but four. The one that remains to be given is equal to
$p_4=\left( -x_1 - x_2 - x_3, \dfrac{ (-x_1 - x_2 - x_3)^2}{2a^2}\right)$.

\bigskip
Now suppose that $\mathcal{M}$ and $\mathcal{N}$ have only an intersection point and an inner tangent point. If we denote $P_1$ and $P_3$ such points such that $\mathcal{M}$ meets $\mathcal{N}$ with multiplicity 3 at $P_3$ and with multiplicity 1 at $P_1$, then, according to the above computations, the values of $P_3$ and $P_1$ can be:

$P_3=p_1=p_2=p_3$ and $P_1=p_4$ with $P_3\neq P_1$. In this case, $x_1=x_2=x_3\neq 0$ and so $-a^2 - x_1^2/a^2$ is triple root of $f$ less than $-a^2$.

$P_3=p_1=p_2=p_4$ and $P_1=p_3$ with $P_3\neq P_1$. In this case, $x_1=x_2$, $x_3 = -3x_2\neq 0$ and so, $-a^2 - x_2^2/a^2$ is triple root of $f$ less than $-a^2$.

$P_3=p_1=p_3=p_4$ and $P_1=p_2$ with $P_3\neq P_1$. In this case, $x_1=x_3,\,x_2=-3x_3\neq 0$ and so, $-a^2 - x_3^2/a^2$ is triple root of $f$ less than $-a^2$.

$P_3=p_2=p_3=p_4$ and $P_1=p_1$ with $P_3\neq P_1$. In this case, $x_2=x_3$, $x_1=-3x_3\neq 0$ and so, $-a^2 - x_3^2/a^2$ is triple root of $f$ less than $-a^2$.

Therefore, in all the cases we find a triple root of $f$ less than $-a^2$.


\bigskip
Going the other way, suppose that $f(\lambda) = 0$ has a triple root, less than $-a^2$. Obviously that implies that $r_1=r_2=r_3, r_1<-a^2$, such that each solution of this system gives rise to two different points of intersection, one of multiplicity three and the other of multiplicity one. 

Namely, the set of solutions of the system of equations $r_1=r_2=r_3$ with $r_1<-a^2$ is described in terms of the variable $x_3$ as follows,
\begin{eqnarray*} 
 x_1 = x_3, & x_2 = x_3,& x_3  \neq 0,\\  
 x_1 = -\dfrac{x_3}{3}, &x_2 = -\dfrac{x_3}{3},& x_3 \neq 0,\\ 
 x_1 = -3x_3, &x_2 = x_3,& x_3 \neq 0, \\ 
 x_1= x_3, &x_2 = -3x_3,& x_3\neq 0;
 \end{eqnarray*}
 Hence, 
 \begin{itemize}
 \item For the solution  $x_1 = x_3, x_2 = x_3, x_3 \neq 0$, the intersection point of multiplicity three is 
 $\left(x_3,  \frac{x_3^2}{2a^2}\right)$ and the one of multiplicity one is $\left(-3x_3,  \frac{9x_3^2}{2a^2}  \right)$.
 
  \item For the solution $x_1 = -\dfrac{x_3}{3}, x_2 = -\dfrac{x_3}{3}, x_3 \neq 0$, the intersection point of multiplicity three is 
   $\left( -\frac{x_3}{3},\frac{  x_3^2}{18a^2}   \right)$ and the one of multiplicity one is $\left( x_3,  \frac{ x_3^2}{2a^2} \right)$.
 
  \item For the solution $x_1 = -3x_3, x_2 = x_3, x_3 \neq 0$, the intersection point of multiplicity three is 
   $\left(   x_3,\frac{x_3^2}{2a^2} \right)$ and the one of multiplicity one is $\left( -3x_3, \frac{9x_3^2}{2a^2}\right)$.

  \item For the solution  $x_1 = x_3, x_2 = -3x_3, x_3\neq 0$, the intersection point of multiplicity three is 
   $\left(  x_3, \frac{x_3^2}{2a^2}  \right)$ and the one of multiplicity one is $\left(-3x_3,  \frac{ 9x_3^2}{2a^2} \right)$.

 \end{itemize}
 
Then it is proved that $\mathcal{M}$ and $\mathcal{N}$ have only an intersection point and an inner tangent point.

\end{proof}

\begin{remark}\begin{enumerate}
\item 
We should add that Points 6., 7. and 8. of Proposition \ref{paracirc} are slightly different from their respective points of Proposition 2.2 in \cite{Liu-Chen:2004} and  the reasons are as follows. On the one hand, Point 8. in \cite{Liu-Chen:2004}  is included in Point 6, and on the other hand, Point 7. in \cite{Liu-Chen:2004}  does not include the case when $-a^2$ is a triple root of $f(\lambda) = 0$. It is easy to see that this happens if and only if $x_c=0, a^2=y_c=\delta$ and the intersection inner point is equal to $(0,0)$; in fact, this case corresponds to the fact that the circle $\mathcal{M}$ is the osculating circle of the parabola at the parabola's vertex $(0,0)$ (see Exercise 3.2.7 in \cite{Univconic}).
\item  In Case 9, the circle $\mathcal{M}$ is the osculating circle of the parabola at the tangent point (see Chapter 3 in \cite{Univconic}). 
\end{enumerate}

\end{remark}
This proposition is the reason why we say that we have reduced the problem at hand into a Quantifier Elimination problem since finally we are looking for the conditions the $c_i$ or $c_i'$ must verify in order that the real roots of $f(\lambda)$ have a prescribed behaviour. In this way, next theorem reformulates the previous proposition in terms of the sign of $\Delta$ and the signs of the coefficients of $f(\lambda)$ and $g(\lambda)$, with the great advantage of not having to calculate the roots of the characteristic polynomial. Although the closed formulas for roots of cubic polynomials (Cardano's Formulae) are known, numerically speaking it is much better to avoid their calculation if it is not absolutely necessary. 

\newpage

\begin{theorem}\label{formulaCanonical}\hskip 0pt\\
Consider the parabola $\mathcal{N}$ and the circle $\mathcal{M}$, as above.
\begin{enumerate}
\item  $\mathcal{M}$ and $\mathcal{N}$  are separated iff $\Delta >0$, and $c_1 >0$ or $c_2 >0$.
\item  $\mathcal{M}$ and $\mathcal{N}$  are externally tangent iff $\Delta =0$, and $c_1 >0$ or $c_2 >0$.

\item  $\mathcal{M}$ is inside $\mathcal{N}$ iff either $\Delta >0$, $c_1 <0$, $c_2<0,$ and
the number of sign changes in $-c_0$, $c_1'$, $-c_2'$, $c_3'$ is 1; or $c_3'=0$, $c_2'=0$, $a^2>\delta$.

\item  $\mathcal{M}$ and $\mathcal{N}$ have only two intersection points iff $\Delta <0$.

\item  $\mathcal{M}$ and $\mathcal{N}$ have four points of intersection iff $\Delta >0, c_2'< 0, c_1'<0 . $

\item  $\mathcal{M}$ and $\mathcal{N}$ have two points of intersection and an inner tangent point iff $a^2 \leq \delta$,
$\Delta=0$, $\Delta'\neq 0$, $c_2' <0, c_1' <0$.

\item  $\mathcal{M}$ and $\mathcal{N}$ have only an inner tangent point iff
\begin{itemize}
\item either $\Delta=0$, $c_1 <0,c_2 <0$, $c_3'<0$, and $c_2' >0$ or $c_1'>0$; 
\item or $\Delta=0$, $c_1 <0,c_2 <0$, $c_3' =0$, $c_1' >0$, $c_2' <0$.  
\item or $c_1'=c_2'=c_3'=0.$

\end{itemize}

\item  $\mathcal{M}$ and $\mathcal{N}$ have two inner tangent points iff $a^2<\delta$,
$c_3' =0$, and $c_2'=0$.

 \item  $\mathcal{M}$ and $\mathcal{N}$ have only an intersection point and an inner tangent point iff $\Delta=0$, $\Delta'=0$, $c_1'<0$ and $c_3'<0$. 
\end{enumerate}
\end{theorem}

Comparing Theorem \ref{formulaCanonical} and results of Section III in \cite{Min-Chen:2008}, one can see that the number of conditions have been reduced in Points 3., 5., 6. and 8., and Point 7. has been completed.
\subsection{Relative position of a parabola and an ellipse}\label{conics}
Assume now that the ellipse $\mathcal{M}$ and the parabola $\mathcal{N}$ are defined, respectively, by the symmetric matrices (not in canonical form)
$$
{\bf M}=\begin{pmatrix}
   a_{11} & a_{12} & a_{13} \\
   a_{12} & a_{22} & a_{23} \\
   a_{13} & a_{23} & a_{33}
\end{pmatrix},\quad
{\bf N}=\begin{pmatrix}
   b_{11} & b_{12} & b_{13} \\
   b_{12} & b_{22} & b_{23} \\
   b_{13} & b_{23} & b_{33}
\end{pmatrix}.
$$
We aim to characterize their relative position directly from their characteristic equation, as in \cite{Min-Chen:2008}, Section IV,  although our results are different from the results in \cite{Min-Chen:2008}.

Following the definitions introduced in Section \ref{cheq}, since $\mathcal{M}$ is a real ellipse, we have $T_1>0$ and $(a_{11}+a_{22}) L_3<0$. Since $\mathcal{N}$ is a parabola,  $T_2 =0$. Assume that
$$
a_{11}>0, a_{22}>0, L_3<0, L_0<0,
$$
and the interior of the ellipse $\mathcal{M}$ is defined by $\mathcal{X}M\mathcal{X}^{T}<0$. Recall that
$$
F(\lambda)=\det(\lambda {\bf N}+ {\bf M})= L_0 \lambda^3+ L_1 \lambda^2+  L_2 \lambda +L_3.
$$

Since $M_2$ is symmetric, it is real diagonalizable and then, there exist an orthogonal matrix $P$ and a diagonal matrix $D$ such that $P^tM_2P=D={\rm diag}(b_0,b_2)$, $D$ defined by the eigenvalues of $M_2$,
$$
b_0=\dfrac{{\rm Trace }\, M_2 +\sqrt{\nu_a}}{2},\;\,
b_2=\dfrac{{\rm Trace }\, M_2 -\sqrt{\nu_a}}{2},\;\, \nu_a = (a_{11}-a_{22})^2 +4a_{12}^2,
%
%
%
$$
and the columns of $P$ by the corresponding orthonormal eigenvectors.

Hence the matrix
$$
\left(\begin{array}{cc}P^t&  \\ & 1\end{array}\right)
$$
defines an isometry such that
$$
{\cal X}{\bf M}{\cal X}^{t}={\cal X}_1{\bf B}{\cal X}_1^{t}=b_0X_1^2+b_2Y_1^2+\cdots,
$$
with
$${\cal X}_1^{t}=\left(\begin{array}{cc}P^t&  \\ & 1\end{array}\right){\cal X}^t,\quad {\bf B}=\left(\begin{array}{cc}P^t&  \\ & 1\end{array}\right)\bf M\left(\begin{array}{cc}P&  \\ & 1\end{array}\right) = \begin{pmatrix}
   b_{0} & 0 & b_{3} \\
  0 & b_{2} & b_{4} \\
   b_{3} & b_{4} & b_{5}
\end{pmatrix},$$

$$b_{5}=a_{33}, \, (b_3,b_4)=(a_{31},a_{32})P .$$

Since this transformation is defined by an isometry, it follows that $b_0 b_2>0$ and $(b_0+b_2)\det {\bf B}=(b_0+b_2)\det {\bf M}<0$, and, consequently, both $b_0, b_2$ must be positive. Next, consider another affine transformation defined by the diagonal matrix ${\rm diag}(\sqrt{b_0},\sqrt{b_2},1)$ and, after completing the square, we obtain
\begin{align*}
 {\cal X}{\bf M}{\cal X}^{t}&={\cal X}_1{\bf B}{\cal X}_1^{t}\\
 &={\cal X}_2\,{\rm diag}(1/\sqrt{b_0},1/\sqrt{b_2},1) \,{\bf B}\,{\rm diag}(1/\sqrt{b_0},1/\sqrt{b_2},1) {\cal X}_2^{t}\\ &=(X_2-X_c)^2+(Y_2-Y_c)^2+\frac{L_3}{T_1},
\end{align*}
with
$${\cal X}_2^t = {\rm diag}(\sqrt{b_0},\sqrt{b_2},1){\cal X}_1^t, \quad X_c=-\frac{b_3}{\sqrt{b_0}},\quad Y_c=-\frac{b_4}{\sqrt{b_2}}\, .
$$
Next, we denote by $\mathcal{M}_1$ the circle
$$
\mathcal{M}_1 : (X_2-X_c)^2 +(Y_2-Y_c)^2 = \frac{-L_3}{T_1},
$$
and by $\mathcal{S}_1$ the affine transformation which transforms $\mathcal{M}$ into $\mathcal{M}_1$, defined by the matrix
$$
{\bf A}= {\rm diag}(\sqrt{b_0},\sqrt{b_2},1) \left(\begin{array}{cc}P^t&  \\ & 1\end{array}\right)\, .
$$
\medskip
In parallel, the parabola $\mathcal{N}$ is transformed by $\mathcal{S}_1$ into $\mathcal{N}_1$,
$$
\mathcal{N}_1  :d_{11} X_2^2 + d_{22}Y_2^2+2 d_{12}X_2Y_2+2d_{13}X_2+2d_{23}Y_2 +d_{33} =0.
$$
As a consequence, if ${\bf N_1}$ and ${\bf M_1}$ are the matrices which define $\mathcal{N}_1$ and $\mathcal{M}_1$ respectively, the characteristic equation $F(\lambda)$ can be written as
$$
F(\lambda)= (\det \bf A)^2  \det (\lambda {\bf N_1}+ {\bf M_1}).
$$
The next step is to obtain the canonical form of $\mathcal{N}_1$. Under another affine transformation $\mathcal{S}_2$ (rigid motion), we obtain
$$
\mathcal{M}_2  : (x'-x'_c)+(y'-y'_c)^2 = \dfrac{-L_3}{T_1},\quad \mathcal{N}_2  :E_1 x'^2 -2 E_2 y' =0,$$
with
$$
E_1 = d_{11}+d_{22}=\dfrac{T}{T_1},
\quad
E_2=\dfrac{\sqrt{ -L_0 } }{\sqrt{ T }} .
$$

By dividing the equation of $\mathcal{N}_2$ by $E_2$, we obtain the same form as in Section~\ref{sec:canonical}, with ${a^2}={E_2}/{E_1}$.
Thus, if ${\bf \hat{N}_2}$ and ${\bf \hat{M}_2}$ denote the matrices which define $\mathcal{N}_2$ and $\mathcal{M}_2$ respectively,
we have
$$f(\lambda) =
\det \left( \lambda\dfrac {\bf \hat{N}_2}{E_2}+ {\bf \hat{M}_2}\right)=c_0 \lambda^3+ c_1 \lambda^2+  c_2 \lambda +c_3,
$$
with
$$
c_0=(\det {\bf A})^{-2} \dfrac{L_0}{E_2^3}=- \dfrac{1}{a^2} ,  \,c_1=(\det {\bf A})^{-2}  \dfrac{L_1}{E_2^2} , \, c_2=(\det {\bf A})^{-2} \dfrac{ L_2}{E_2},$$
and $c_3=(\det {\bf A})^{-2} L_3$.

Note that, in order to apply Theorem \ref{formulaCanonical}, we have ${\rm sign}(c_i)={\rm sign}(L_i)$ ($i=0,\ldots,3$).

The coefficients of $g(\lambda)=f(\lambda-a^2)=c_0'\lambda^3+c_1'\lambda^2+c_2'\lambda+c_3'$ will be
\begin{align*}
c_0'&= c_0, \quad c_1' =\left(\det {\bf A}\right)^{-2} \left(-3 L_0 \dfrac{a^2}{E_2^3} + \dfrac{L_1}{E_2^2}\right), \\
  c_2'&=(\det {\bf A})^{-2}\left(3 L_0  \dfrac{ a^4}{E_2^3} - 2 L_1   \dfrac{a^2}{E_2^2} +  \dfrac{L_2 }{E_2}\right), \\
  c_3'&=(\det {\bf A})^{-2}\left(-L_0   \dfrac{a^6}{E_2^3} + L_1   \dfrac{a^4}{E_2^2} - L_2   \dfrac{a^2}{E_2} +  L_3\right).
\end{align*}

Since $T>0$, $E_2>0$, $a^2=E_2/E_1$ and $E_1=T/T_1$, we infer that $c_1'$ has the same sign as
$$I_5= -3 L_0 T_1 + L_1 T,$$
$c_2'$ has the same sign as $$I_4= 3 L_0 T_1^2 - 2L_1 T T_1+L_2 T^2 $$
and $c_3'$ has the same sign as
$$I_3= -L_0 T_1^3+L_1 T T_1^2-L_2 T^2 T_1+L_3 T^3 . $$
Let
$$
I_2= -T_1^3L_0+L_3T^3.
$$

After all these calculations, we are now ready to state a version of Theorem \ref{formulaCanonical} when the parabola and ellipse are not given in canonical form.


\begin{theorem}\label{ep}\hskip 0pt\\
The relationship between the ellipse $\mathcal{M}$ and the parabola $\mathcal{N}$ are as follows:
\begin{enumerate}
\item $\mathcal{M}$ and $\mathcal{N}$ are separated iff $\Delta>0$ and $(L_1>0 \;\hbox{or}\; L_2>0)$.
\item $\mathcal{M}$ and $\mathcal{N}$  are externally tangent iff $\Delta=0$ and $(L_1>0  \;\hbox{or}\;  L_2>0)$. 
\item  $\mathcal{M}$  is inside $\mathcal{N}$ iff   
$\{\Delta>0 ,  L_1<0 ,  L_2<0 ,  I_4>0 \}$
or
$\{\Delta>0, L_1<0, L_2<0, I_5\geq 0,I_4\leq 0, I_3<0\}$,
or
$\{ I_3 =0,I_4 =0,I_5  >0\}$.

\item $\mathcal{M}$ and $\mathcal{N}$  have only two points of intersection iff $\Delta<0$.
\item $\mathcal{M}$ and $\mathcal{N}$  have four points of intersection iff $\Delta>0$, $I_4 <0$ and $I_5 <0$.
\item $\mathcal{M}$ and $\mathcal{N}$ have one inner tangent point and two points of intersection points iff $I_2\leq 0, \Delta=0, \Delta' \neq 0, I_4<0,I_5<0$.
\item $\mathcal{M}$ and $\mathcal{N}$  have only one inner tangent point iff 
{ $\Delta=0$, $L_1<0$, $L_2<0$, $I_3<0$ and $(I_4>0 \;\hbox{or}\; I_5>0 )$}, 
or 
$\{\Delta=0$, $L_1<0$, $L_2<0$, $I_3=0$, $I_5>0$, $I_4<0 \}$,
or 
$\{I_3=0$, $I_5=0$, $I_4=0\}$.

\item $\mathcal{M}$ and $\mathcal{N}$ have only two inner points of tangency points iff $I_3 =0$, $I_4 =0$ and $I_5<0$.

 \item  $\mathcal{M}$ and $\mathcal{N}$ have only an intersection point and an inner tangent point iff $\Delta=0$, $\Delta'=0$, $I_3<0$ and $I_5<0$.

\end{enumerate}
\end{theorem}

We recover here the same formulae introduced in \cite{AET,CDFG} for characterizing when two ellipses are separated or they are externally tangent. This is not a surprise since the root pattern for these two geometric configurations is the same as in the case we have considered here. The next picture illustrates all relations described in Theorem \ref{ep}.
\begin{figure}[H]
\begin{center}
 \fbox{\includegraphics[scale=.17]{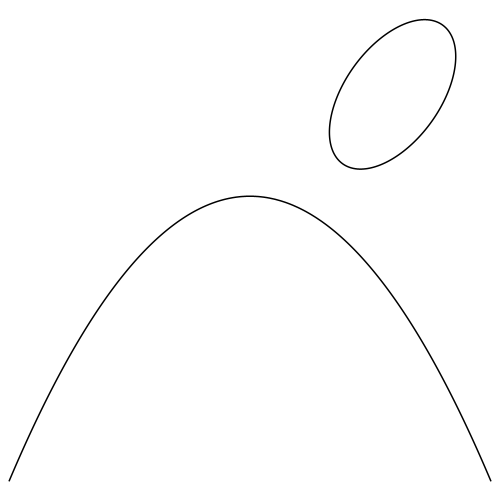} {\bf 1}} \hskip 3pt
  \fbox{\includegraphics[scale=.17]{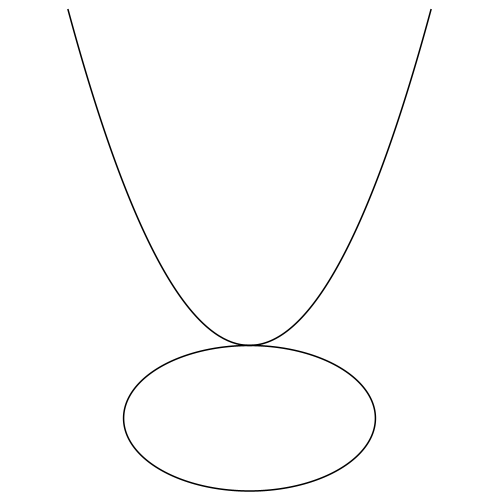} {\bf 2}} \hskip 3pt
 \fbox{\includegraphics[scale=.17]{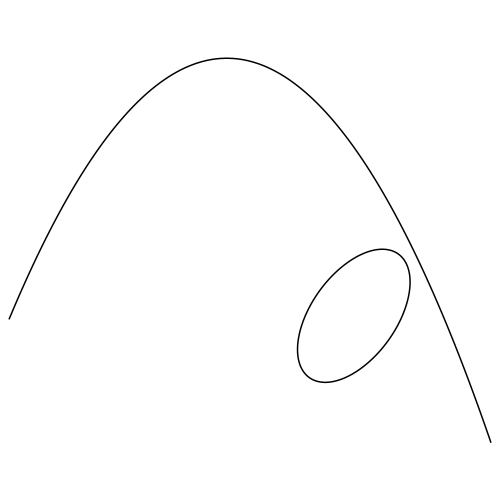} {\bf 3}} \hskip 3pt
 \fbox{\includegraphics[scale=.17]{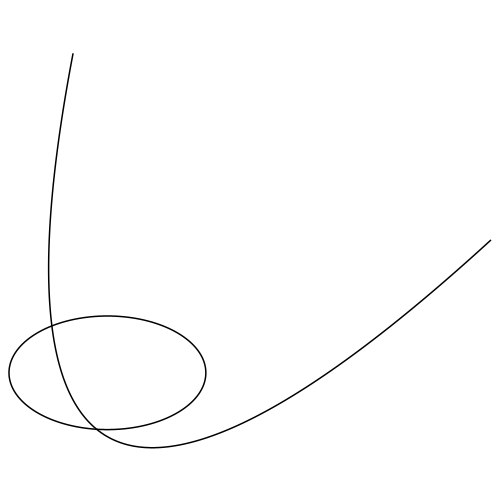} {\bf 4}} \\
 \vspace{2pt}
  \fbox{\includegraphics[scale=.17]{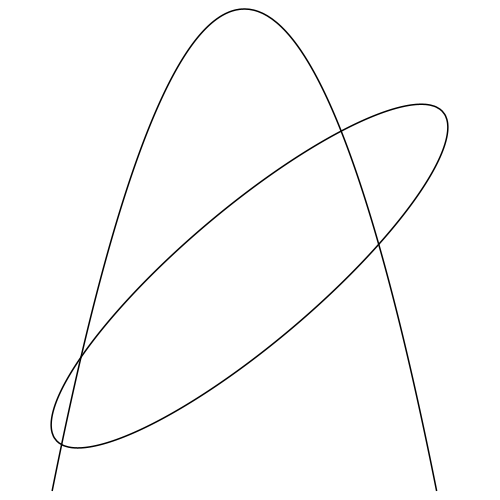} {\bf 5}} \hskip 3pt
   \fbox{\includegraphics[scale=.17]{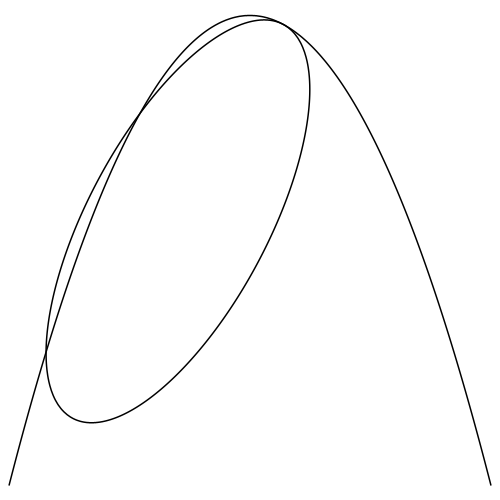} {\bf 6}} \hskip 3pt
 \fbox{\includegraphics[scale=.17]{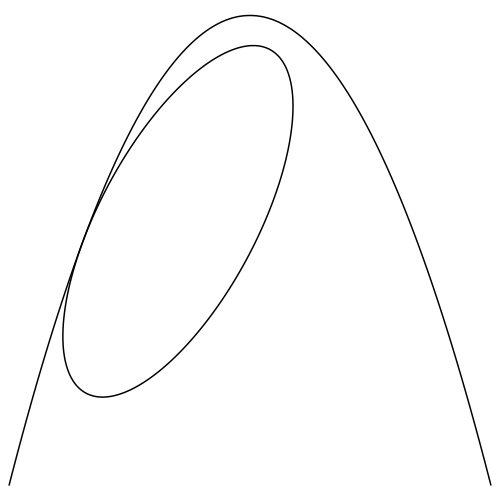} {\bf 7}} \hskip 3pt
  \fbox{\includegraphics[scale=.17]{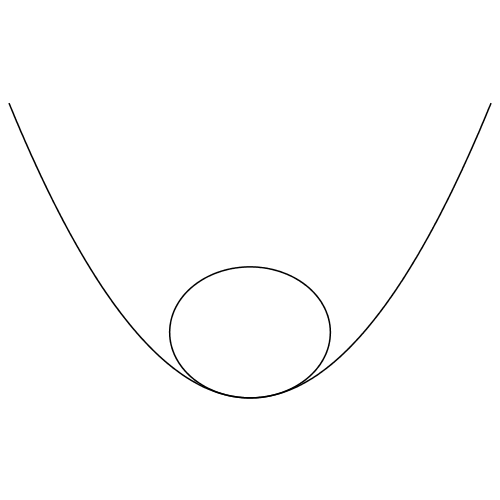} {\bf 7}} \hskip 3pt 
  \fbox{\includegraphics[scale=.17]{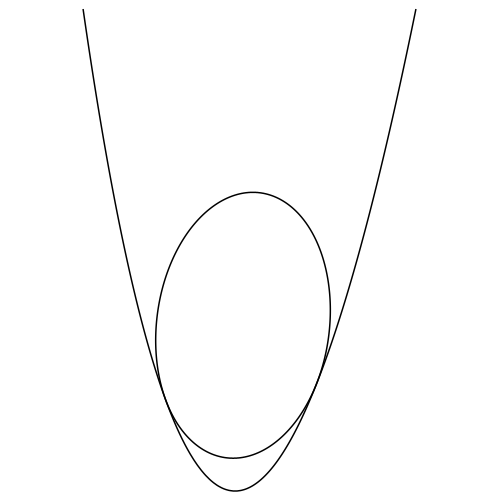} {\bf 8}}
   \fbox{\includegraphics[scale=.17]{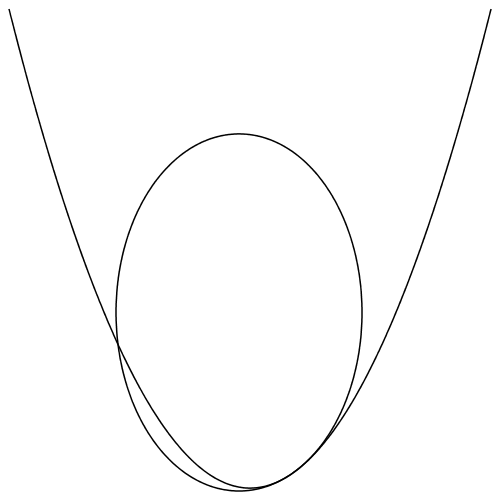} {\bf 9}} 
\end{center}
\caption{Ellipse and Parabola - Theorem \ref{ep}}
\end{figure}

\begin{remark}
Theorem 1 in  \cite{Min-Chen:2008} is apparently similar to Theorem \ref{ep}, however, our definitions of $I_3$, $I_4$ and $I_5$ are different and the number of conditions have been reduced by Theorem \ref{formulaCanonical}.
\end{remark}

\section{Hyperbola and ellipse}

Similarly to the treatment of the parabola and ellipse case, we will first consider the relative position of a hyperbola in canonical form and a circle. Then we will consider any hyperbola and any ellipse, find the transformations that convert the hyperbola into canonical form and the ellipse into a circle, and rephrase the main results accordingly. Our approach could be used for example in \cite{Ding}, to better differentiate between the relative position of these two conics.

\subsection{Relative position. Canonical form}\label{sec:canonicalEH}

Given a hyperbola in canonical form $\mathcal{NH} : {\cal X}{\bf Nh}{\cal X}^{t}=\dfrac{x^2}{a^2}-\dfrac{y^2}{b^2}+1=0$, with
$$
{\bf Nh}=\left( \begin {array}{ccc} a^{-2}&0&0\\ \noalign{\medskip}0&-{b}^{-2}&0
\\ \noalign{\medskip}0&0&1\end {array} \right),
$$
$a>0,b>0$, and a circle $\mathcal{M}:{\cal X}{\bf M}{\cal X}^{t}=0$, with

$$
{\bf M}=\left( \begin {array}{ccc} 1&0&-{\it x_c}\\ \noalign{\medskip}0&1&-{
\it y_c}\\ \noalign{\medskip}-{\it x_c}&-{\it y_c}&-{\delta}^{2}+{{\it x_c
}}^{2}+{{\it y_c}}^{2}\end {array} \right),
$$
$\delta>0$, their characteristic equation $f(\lambda)=\det(\lambda\,{\bf Nh}+{\bf M})$ is

$$
f(\lambda)=a_3\lambda^3+a_2\lambda^2+a_1\lambda+a_0,
$$
with
 \begin{align*}
a_3&=-\dfrac {1}{{a}^{2}{b}^{2}}, \;\; a_2= -\dfrac { \left( {a}^{2}-{b}^{
2}-{\delta}^{2}+{{\it x_c}}^{2}+{{\it y_c}}^{2} \right) }{{
a}^{2}{b}^{2}},  \\ \\
a_1&=\dfrac { {a}^{2}{b}^{2}+{a}^{2}{\delta}^{2}-{a}
^{2}{{\it y_c}}^{2}-{b}^{2}{\delta}^{2}+{b}^{2}{{\it x_c}}^{2} }{{a}^{2}{b}^{2}}\, \text{ and } \, a_0=-\delta^2.
\end{align*}
 %
Observe that
\begin{enumerate}
\item the coefficients $a_3$ and $a_0$ are strictly negative,
\item the characteristic equation $f(\lambda)$ has at least one negative real root.
\end{enumerate}

Following both \cite{Brozos:2018} and \cite{Liu-Chen:2004}, next proposition describes the position relationship between a circle and a hyperbola in terms of the roots of $f(\lambda)$, which will be denoted $\lambda_1, \lambda_2, \lambda_3$. As it happened in Proposition \ref{paracirc}, the last case was missing in our references and that is why it is proved.
\begin{proposition}\label{ch}
The positional relationships between a circle and a hyperbola is as follows:
\begin{enumerate}

\item $\mathcal{NH}$ and $\mathcal{M}$ are separated iff $f(\lambda) = 0$ has two distinct positive roots equal to or less than $b^2$; that is, $\lambda_1<0<\lambda_2<\lambda_3\leq b^2$.

\medskip

\item  $\mathcal{NH}$ and $\mathcal{M}$  have only two points of intersection iff $f(\lambda) = 0$ has two imaginary roots. 


\item \label{tres}

$\mathcal{NH}$ and $\mathcal{M}$  have only an inner tangent point iff $f(\lambda) = 0$ has a negative double root greater than $-a^2$ ($ \lambda_1<-a^2<\lambda_2=\lambda_3<0$), or the three roots are $-a^2$, $-b\delta$, $-b\delta$ with $a^2 \geq b\delta$.

%
%
%
%

%



\item  \label{cuatro}
$\mathcal{NH}$ and $\mathcal{M}$ have only two inner tangent points iff $a^2 < b\delta$ and the roots of $f(\lambda) $ are $-a^2$, $-a^2$, $-\dfrac{b^2\delta^2}{a^2}$.



\item \label{cinco}
%
$\mathcal{NH}$ and $\mathcal{M}$  have two points of intersection and an inner tangent point iff  $f(\lambda) = 0$ has a negative double root and all roots are less than or equal to $ -a^2$ with $a^2 < b\delta$. If a root is $-a^2$, then the three roots are $-a^2$,$-b\delta$, $-b\delta$.  

%
%
%

\item  $\mathcal{NH}$ and $\mathcal{M}$ have only one outer tangent point iff $f(\lambda) = 0$ has a positive double root which is less than $b^2$; that is $\lambda_1<0<\lambda_2=\lambda_3<b^2$. 


\item\label{siete} $\mathcal{NH}$ and $\mathcal{M}$  have two outer tangent points iff
the roots of $f(\lambda) = 0$ are $b^2,b^2$, $-\dfrac{a^2\delta^2}{b^2}$ with $b \leq \delta$. 



\item \label{ocho} 

$\mathcal{NH}$ and $\mathcal{M}$  have two points of intersection and an outer tangent point iff $f(\lambda) $ has a positive double root greater than $b^2$, $b^2< \lambda_1=\lambda_2 $, with $b < \delta$. 



\item \label{nueve}   $\mathcal{NH}$ and $\mathcal{M}$  have four points of intersection iff either $f(\lambda) = 0$ has three distinct negative roots which are not greater than $-a^2$, $ \lambda_1<\lambda_2<\lambda_3\leq-a^2$,  with $a^2<b\delta$ or $f(\lambda)$ has two distinct positive roots not less than $b^2$, $ \lambda_1<0<b^2\leq\lambda_2<\lambda_3$ with $b<\delta$.  

%
%
%
%
%
%

\item\label{diez} $\mathcal{M}$ is inside  $\mathcal{NH}$ iff $f(\lambda) = 0$ has three distinct negative roots, two of which are not less than $-a^2$,( $ \lambda_1<-a^2\leq\lambda_2<\lambda_3 < 0$ or  $ \lambda_1\leq-a^2<\lambda_2<\lambda_3 < 0$), or the three roots are $-a^2$, $-a^2$, $-\dfrac{b^2\delta^2}{a^2}$ with $a^2>b\delta$.

%
%
%
%
%

\item\label{once} $\mathcal{M}$ and $\mathcal{NH}$ have only an intersection point and an inner tangent point iff $f(\lambda) = 0$ has a triple negative root less than $-a^2$.
 
 \end{enumerate}
\end{proposition}

\begin{proof}
As we said previously, we must prove Point \ref{once} (see \cite{Liu-Chen:2004} for the rest).  

We are going to follow the same reasoning as in Proposition \ref{paracirc}. In general, the hyperbola $\mathcal{NH}$ and the circle $\mathcal{M}$ have four points of intersection in $\mathbb{C}^2$, maybe equal depending on the multiplicities. 
By considering the following rational parameterization of the hyperbola, 
$$
x=\frac{a(t^2 - 1)}{2t}, \quad y=\frac{b(t^2 + 1)}{2t},
$$
three of these intersection points are given symbolically as $ p_1=\left( \frac{a(t_1^2 - 1)}{2t_1}, \frac{b(t_1^2 + 1)}{2t_1} \right) $, $ p_2=\left( \frac{a(t_2^2 - 1)}{2t_2}, \frac{b(t_2^2 + 1)}{2t_2} \right)  $ and $ p_3=\left( \frac{a(t_3^2 - 1)}{2t_3}, \frac{b(t_3^2 + 1)}{2t_3} \right)  $. After some computations, another equation for the circle $\mathcal{M}$ is provided by the circle defined by these three points,
$$
\mathcal{M}: x^{2}+y^{2}+d_1\,x+d_2\,y+d_3=0,
$$
with

\begin{eqnarray*} 
d_1&=&\frac{(t_2t_3 - 1)(t_1t_3 - 1)(t_1t_2 - 1)(a^2 + b^2) }{4 t_1t_2at_3},\\ \\
d_2&=&-\frac{(t_2t_3 + 1)(t_1t_3 + 1)(t_1t_2 + 1)(a^2 + b^2)}{4bt_1t_2t_3},\\ \\
d_3&=&\frac{(a^2 + b^2)(t_1^2t_2^2t_3 + t_1^2t_2t_3^2 + t_1t_2^2t_3^2 + t_1 + t_2 + t_3)}{4t_1t_2t_3} + \frac{a^2 - b^2}{2}.
\end{eqnarray*}

Once we have this equation, the characteristic polynomial of $\mathcal{NH}$ and $\mathcal{C}$ can also be expressed  in these terms,
$$
f(\lambda)=\dfrac{-1}{a^2b^2}(\lambda-r_1 )(\lambda-r_2 )(\lambda-r_3 ),$$ 
with 
$$
r_1=  -\dfrac{(t_2t_3 + 1)^2a^2 + (t_2t_3 - 1)^2b^2}{4t_2t_3},\quad r_2= - \dfrac{(t_1t_3 + 1)^2a^2+ (t_1t_3 - 1)^2b^2}{4t_1t_3}, $$ $$ r_3=  - \dfrac{(t_1t_2 + 1)^2a^2 + (t_1t_2 - 1)^2b^2}{4t_1t_2};
$$
Besides, as we mentioned before, the intersection of the circle and the parabola does not consist only of three points but four. The one that remains to be given is equal to
$$p_4=\left(\dfrac{-a(t_1^2t_2^2t_3^2 - 1)}{2t_1t_2t_3},\dfrac{ b(t_1^2t_2^2t_3^2 + 1)}{2t_1t_2t_3}\right).$$

\bigskip
Now suppose that $\mathcal{M}$ and $\mathcal{NH}$ have only an intersection point and an inner tangent point. If we denote $P_1$ and $P_3$ these points, such that $\mathcal{M}$ meets $\mathcal{NH}$ with multiplicity 3 at $P_3$ and with multiplicity 1 at $P_1$, then, according to the above computations, the values of $P_3$ and $P_1$ can be:

$P_3=p_1=p_2=p_3$ and $P_1=p_4$ with $P_3\neq P_1$. In this case, $t_1=t_2=t_3\neq 0$, $t_1^2\neq \pm 1$, and so $  -\dfrac{(t_1^2 + 1)^2a^2 + (t_1^2 - 1)^2b^2}{4t_1^2}$ is triple root of $f$ less than $-a^2$ since $\dfrac{(t_1^2 + 1)^2}{4t_1^2}>1$.

$P_3=p_1=p_2=p_4$ and $P_1=p_3$ with $P_3\neq P_1$. In this case, $t_1 = t_2, t_3 = 1/t_2^3, t_2 \neq 0$,  $t_2^2\neq \pm 1$ and so, $  -\dfrac{(t_2^2 + 1)^2a^2 + (t_2^2 - 1)^2b^2}{4t_2^2}$ is triple root of $f$ less than $-a^2$ since $\dfrac{(t_2^2 + 1)^2}{4t_2^2}>1$.

$P_3=p_1=p_3=p_4$ and $P_1=p_2$ with $P_3\neq P_1$. In this case, $t_1 = t_3, t_2 = 1/t_3^3, t_3 \neq 0$,  $t_3^2\neq \pm 1$ and so, $  -\dfrac{(t_3^2 + 1)^2a^2 + (t_3^2 - 1)^2b^2}{4t_3^2}$ is triple root of $f$ less than $-a^2$ since $\dfrac{(t_3^2 + 1)^2}{4t_3^2}>1$.

$P_3=p_2=p_3=p_4$ and $P_1=p_1$ with $P_3\neq P_1$. In this case, $t_1 = 1/t_3^3, t_2 = t_3, t_3 \neq 0$,  $t_3^2\neq \pm 1$ and so, $  -\dfrac{(t_3^2 + 1)^2a^2 + (t_3^2 - 1)^2b^2}{4t_3^2}$ is triple root of $f$ less than $-a^2$ since $\dfrac{(t_3^2 + 1)^2}{4t_3^2}>1$.

Therefore, in all the cases we find a triple root of $f$ less than $-a^2$.

\bigskip
Going the other way, suppose that $f(\lambda) = 0$ has a triple root, less than $-a^2$. Obviously that implies that $r_1=r_2=r_3, r_1<-a^2$, such that each solution of this system gives rise to two different points of intersection, one of multiplicity three and the other of multiplicity one. 

Namely, the set of solutions of the system of equations $r_1=r_2=r_3$ with $r_1<-a^2$ is described  as follows, 
\begin{eqnarray*} 
 t_1 = t_3, & t_2 = t_3,& t_3  \neq 0,\quad t_3^2  \neq 1,\\  
 t_1 =t_2, &t_3 = \dfrac{1}{t_2^3},& t_2 \neq 0,\quad t_2^2  \neq 1,\\ 
 t_1 = \dfrac{1}{t_3^3}, &t_2 = t_3,& t_3 \neq 0,\quad t_3^2  \neq 1, \\ 
 t_1= t_3, &t_2 = \dfrac{1}{t_3^3},& t_3\neq 0,\quad t_3^2  \neq 1.
 \end{eqnarray*}
 Hence, 
 \begin{itemize}
 \item For the solution  $t_1 = t_2=t_3, t_3 \neq 0$, the intersection point of multiplicity three is 
$ \left( \frac{a(t_3^2 - 1)}{2t_3}, \frac{b(t_3^2 + 1)}{2t_3} \right) $ and the one of multiplicity one is $\left( \frac{-a(t_3^6 - 1)}{2t_3^3}, \frac{b(t_3^6 + 1)}{2t_3^3} \right) $.
 
  \item For the solution $t_1 =t_2, t_3 = \dfrac{1}{t_2^3}, t_2 \neq 0$, the intersection point of multiplicity three is 
   $ \left( \frac{a(t_2^2 - 1)}{2t_2}, \frac{b(t_2^2 + 1)}{2t_2} \right) $ and the one of multiplicity one is $\left( \frac{-a(t_2^6 - 1)}{2t_2^3}, \frac{b(t_2^6 + 1)}{2t_2^3} \right) $.

  \item For the solution $ t_1 = \dfrac{1}{t_3^3},t_2=t_3, t_3 \neq 0$, the intersection point of multiplicity three is 
   $ \left( \frac{a(t_3^2 - 1)}{2t_3}, \frac{b(t_3^2 + 1)}{2t_3} \right) $ and the one of multiplicity one is $\left( \frac{-a(t_3^6 - 1)}{2t_3^3}, \frac{b(t_3^6 + 1)}{2t_3^3} \right) $.
  
  \item For the solution $t_1 =t_3, t_2 = \dfrac{1}{t_3^3}, t_3 \neq 0$, the intersection point of multiplicity three is 
   $ \left( \frac{a(t_3^2 - 1)}{2t_3}, \frac{b(t_3^2 + 1)}{2t_3} \right) $ and the one of multiplicity one is $\left( \frac{-a(t_3^6 - 1)}{2t_3^3}, \frac{b(t_3^6 + 1)}{2t_3^3} \right) $.
 \end{itemize}
 
Then it is proved that $\mathcal{M}$ and $\mathcal{NH}$ have only an intersection point and an inner tangent point.

\end{proof}


Comparing this proposition with Proposition 2.3 in \cite{Liu-Chen:2004}, in addition to the fact that we consider a circle of any radius ($\delta$ may be different from 1), Points 3., 4., 5., 8. and 9. are slightly improved. In Case 3., the circle $\mathcal{M}$ is the osculating circle of the hyperbola at the vertex $(0,b)$ when $a^2=b\delta, x_c=0, y_c=b+\delta$; and at the vertex $(0,-b)$ when $a^2=b\delta, x_c=0, y_c=-b-\delta$. (See Exercise 3.2.6 in \cite{Univconic}). In Case 11, the circle $\mathcal{M}$ is the osculating circle of the hyperbola at the tangent point.


The result explained below aims to avoid calculation of roots, just as it happened with the ellipse and circle, characterizing the relative position between a circle and a hyperbola in terms of coefficients of different polynomials.

Let $g(\lambda)=f(\lambda-a^2)=a_3'\lambda^3+a_2'\lambda^2+a_1'\lambda+a_0'$, with
\begin{eqnarray*}
 a_3'&=&-\dfrac{1}{a^{2} b^{2}} \, ,\\
 a_2'&=& a_2-3a^2a_3 \\
 &=& \frac{2 a^{2}+b^{2}+\delta^{2}-\mathit{x_c}^{2}-\mathit{y_c}^{2}}{b^{2} a^{2}} \, ,\\ 	
 a_1'&=&3a^4a_3 - 2a^2a_2 + a_1 \\
 &=& -\dfrac{a^{4}+b^{2} a^{2}+\delta^{2} a^{2}+b^{2} \delta^{2}-2 a^{2} \mathit{x_c}^{2}-\mathit{y_c}^{2} a^{2}-b^{2} \mathit{x_c}^{2}}{b^{2} a^{2}}\\
 &=&\dfrac{2 a^{2} \mathit{x_c}^{2}+\mathit{y_c}^{2} a^{2}+b^{2} \mathit{x_c}^{2}-(a^2 + \delta^2)(a^2 + b^2) }{b^{2} a^{2}}  \, ,\\
 a_0'&=& -a^6a_3 + a^4a_2 - a^2a_1 + a_0\\
&=& -\dfrac{\mathit{x_c}^{2} \left(b^{2}+a^{2}\right)}{b^{2}} \, .
\end{eqnarray*}
\medskip
In a similar way,
let $q(\lambda)=f(\lambda+b^2)=a_3''\lambda^3+a_2''\lambda^2+a_1''\lambda+a_0''$, with
\begin{eqnarray*}
a_3''&=&-\frac{1}{a^{2} b^{2}}  \, ,\\
a_2''&=&a_2+3a_3b^2 \\
&=&-\frac{\mathit{x_c}^{2}+\mathit{y_c}^{2}+a^{2}+2 b^{2}-\delta^{2}}{a^{2} b^{2}}  \, ,\\
a_1''&= &3a_3b^4 + 2a_2b^2 + a_1\\
&=&-\frac{a^{2} b^{2}-a^{2} \delta^{2}+b^{4}-b^{2} \delta^{2}+b^{2} \mathit{x_c}^{2}+a^{2} \mathit{y_c}^{2}+2 b^{2} \mathit{y_c}^{2}}{a^{2} b^{2}}\\
&=&-\frac{(b^2 - \delta^2)(a^2 + b^2)+b^{2} \mathit{x_c}^{2}+a^{2} \mathit{y_c}^{2}+2 b^{2} \mathit{y_c}^{2}}{a^{2} b^{2}}  \, , \\
a_0''&=&-\frac{\mathit{y_c}^{2} \left(a^{2}+b^{2}\right)}{a^2} \, .
\end{eqnarray*}
From the coefficients of the polynomials $g(\lambda)$ and $q(\lambda)$, we can deduce the following observations.

\begin{enumerate}
\item  Since $b^2+a^2>0$,  $-a^2$ is a root of $f(\lambda)$ iff  $x_c= 0$; besides if $x_c\neq 0$, $g(\lambda)$ must have a negative root, that is, $f(\lambda)$ must have a negative root less than $-a^2$.

\item $f(\lambda)=-\dfrac{( \lambda+a^2 )( \lambda+b\delta )^2 }{a^2b^2}$ iff  $x_c= 0$ and $|y_c|=b+\delta$.

\item  Since $b^2+a^2>0$, $b^2$ is a root of $f(\lambda)$ iff  $y_c= 0$; besides if $y_c\neq 0$, $q(\lambda)$ must have a negative root, that is, $f(\lambda)$ must have a root less than $b^2$.
\end{enumerate}



Thus, the following proposition characterizes the relative position between a circle and a hyperbola in terms of the coefficients of $f(\lambda), g(\lambda)$ and $q(\lambda)$.

\begin{proposition}\label{ch2}
The positional relationships between a circle and a hyperbola are as follows:
\end{proposition}
\begin{enumerate}

\item $\mathcal{NH}$ and $\mathcal{M}$ are separated iff $\Delta>0,{\bf Var}(f)=2,{\bf Var}(q)=0$.

\item  $\mathcal{NH}$ and $\mathcal{M}$  have only two points of intersection iff $\Delta<0$.

\item $\mathcal{NH}$ and $\mathcal{M}$  have only an inner tangent point iff either  $\Delta=0$, ${\bf Var}(f)=0$, ${\bf Var}(g)=2$ and $a_0'\neq  0$, 
or  $a_0'=0$, $a_2'\geq 0$ and  $| y_c |=b+\delta$.




\item  $\mathcal{NH}$ and $\mathcal{M}$ have only two inner tangent points iff $a_0'=0$, $a_1'=0$ and $a_2'<0$. 


\item  $\mathcal{NH}$ and $\mathcal{M}$ have two points of intersection and an inner tangent point iff \begin{itemize}\item either $\Delta=0$, $\Delta'\neq 0$, ${\bf Var}(g)=0$ and $a'_0\neq 0$, \item or $ a'_0=0$, $|y_c|=b+\delta$ and $a'_2<0$.\end{itemize}



\item  $\mathcal{NH}$ and $\mathcal{M}$ have only one outer tangent point iff $\Delta=0$, ${\bf Var}(f)=2$, ${\bf Var}(q)=0$, $y_c\neq 0$. 

\item $\mathcal{NH}$ and $\mathcal{M}$  have two outer tangent points iff $  a_0''=0  ,  a_1''=0 $.


\item   $\mathcal{NH}$ and $\mathcal{M}$  have two points of intersection and an outer tangent point if $\Delta=0$, $  {\bf Var}(q)=2, a''_0\neq 0$.


\item $\mathcal{NH}$ and $\mathcal{M}$  have four points of intersection iff $\Delta>0$ and either ${\bf Var}(g)= 0$; or ${\bf Var}(q)> 0$.
%
%

\item $\mathcal{M}$ is inside  $\mathcal{NH}$ iff ${\bf Var}(f)= 0$ and, either $\Delta>0$,  $a_0'\neq 0$, ${\bf Var}(g)=2$; or $\Delta>0$, $a_0'=0$, $a_1'\neq 0$ and ${\bf Var}(g)>0$; or $a_0'=0$, $a_1'=0$ and $a_2'>0$.


\item  
$\mathcal{M}$ and $\mathcal{NH}$ have only an intersection point and an inner tangent point iff 
$\Delta=0$, $\Delta'=0$,  $a_0'<0$ and $a_2'<0$. 

 \end{enumerate}

 \begin{proof}
 Points 1., 2., 6., 11. are quite obvious because of Lemma \ref{ellipsesdisc} and Proposition \ref{realr}.
 \begin{enumerate}
\item[3.] By \ref{tres} in Proposition  \ref{ch}, $\mathcal{NH}$ and $\mathcal{M}$  have only an inner tangent point iff $f(\lambda) = 0$ has a negative double root greater than $-a^2$ ($ \lambda_1<-a^2<\lambda_2=\lambda_3<0$) , or the three roots are $-a^2$, $-b\delta$, $-b\delta$ with $a^2 \geq b\delta$.

It is evident that $f(\lambda) = 0$ has a negative double root greater than $-a^2$ iff $\Delta=0$, $a'_0\neq 0$, ${\bf Var}(f)=0$ and ${\bf Var}(g)=2$.

On the other hand, we claim that the three roots are $-a^2$, $-b\delta$, $-b\delta$ with $a^2 \geq b\delta$ iff  $a'_0=0$, $|y_c|=b+\delta$ and $a'_2\geq 0$.

In fact,  $a'_0= 0$,  $a'_2\geq 0$ and $|y_c|=b + \delta$ iff
\begin{align}
f(\lambda)&=-{\dfrac { \left( {a}^{2}+\lambda \right)  \left( b\delta+\lambda \right) ^{2}}{{a}^{2}{b}^{2}}}\, ,
\\
g(\lambda)&=-{\dfrac {{\lambda}^{3}}{{a}^{2}{b}^{2}}}+2\,{\dfrac {
 \left( {a}^{2}-b\delta \right) \lambda^{2}}{{a}^{2}{b}^{2}}}-{\dfrac { \left( {a}
^{2}-b\delta \right) ^{2}\lambda}{{a}^{2}{b}^{2}}}\, ,
\end{align}
with $a^2 \geq b\delta$.


To better understand what happens, with $x_c=0$ $(a_0'=0)$, by dividing $f(\lambda)$ by $\lambda+a^2$, we obtain a factor of degree 2 whose discriminant is equal to $(b - \delta + y_c) (b + \delta + y_c)(b - \delta - y_c)(b + \delta - y_c)$, and $-b\delta$ is a double root if only if $|y_c|=b + \delta$.

 \item[4.] By \ref{cuatro} in Proposition  \ref{ch}, $\mathcal{NH}$ and $\mathcal{M}$ have only two inner tangent points iff  $a^2 < b\delta$ and the roots of $f(\lambda) $ are $-a^2$, $-a^2$, $-\dfrac{b^2\delta^2}{a^2}$. This is true when $g(\lambda)=-\dfrac{\lambda^{3}}{a^2 b^2}+\dfrac{\lambda^{2} \left(a^{4}-b^{2} \delta^{2}\right)}{b^{2} a^{4}}$, with $a_0'=a_1'=0$ and $a_2'<0$.

\item[5.] First suppose that $a_0'\neq 0$, so that $-a^2$ is not a root of $f(\lambda)$. Then, by \ref{cinco} in Proposition \ref{ch}, $\mathcal{NH}$ and $\mathcal{M}$ have two points of intersection and an inner tangent point iff $f(\lambda) = 0$ has a negative double root and all roots are less than $-a^2$ with $a^2<b\delta$, which is equivalent to $\Delta=0$, $\Delta'\neq 0$ and ${\bf Var}(g)=0$. Notice that $a^2 b^2 \delta^2=|\lambda_1\lambda_2\lambda_3| $, therefore $a^6 < a^2 b^2 \delta^2$, which implies that $a^2 < b\delta$.

Next, suppose $a_0'= 0$. Then, $\mathcal{NH}$ and $\mathcal{M}$ have two points of intersection and an inner tangent point iff the roots are $-a^2,-b\delta,-b\delta$ with $a^2<b\delta$. As we mention in the point 3., the roots $-a^2,-b\delta,-b\delta$ iff $a_0'= 0$ and $|y_c|=b+\delta$ with
$$g(\lambda)=-{\dfrac {{\lambda}^{3}}{{a}^{2}{b}^{2}}}+2\,{\dfrac {
 \left( {a}^{2}-b\delta \right) \lambda^{2}}{a^2b^2}}-{\dfrac { \left( {a}
^{2}-b\delta \right) ^{2}\lambda}{{a}^{2}{b}^{2}}},$$
hence the condition $a^2<b\delta$ is equivalent to $a_2'<0$.

\item[7.] By \ref{siete} in Proposition  \ref{ch}, $\mathcal{NH}$ and $\mathcal{M}$  have two outer tangent points iff
the roots of $f(\lambda) = 0$ are $b^2, b^2$, $-\dfrac{a^2\delta^2}{b^2}$ with $b \leq \delta$. This is equivalent to  $a_0''=a_1''=0$  with $b \leq \delta$.
However, as for the condition $b \leq \delta$, notice that when $a_0''=0$, we have
$$q(\lambda)=-\frac{\lambda^{3}}{a^{2} b^{2}}-\frac{\left(a^{2}+2 b^{2}-\delta^{2}+\mathit{x_c}^{2}\right) \lambda^{2}}{a^{2} b^{2}}-\frac{\left(a^{2} b^{2}-a^{2} \delta^{2}+b^{4}-b^{2} \delta^{2}+b^{2} \mathit{x_c}^{2}\right) \lambda}{a^{2} b^{2}},$$ and $a_1''=0$ leads to $x_c^2=-\dfrac{\left(b^2 -\delta^2  \right) \left(a^{2}+b^{2}\right)}{b^{2}}$, and hence $b\leq \delta$ such that the condition $a_0''=a_1''=0$ is sufficient.

\item[8.] By \ref{ocho} in Proposition  \ref{ch}, $\mathcal{NH}$ and $\mathcal{M}$ have two points of intersection and an outer tangent point iff $f(\lambda) $ has a positive double root greater than $b^2$ with $b < \delta$. Clearly, $\Delta=0$, $  {\bf Var}(q)=2$ and $a''_0\neq 0$ are necessary conditions. As for the condition $b < \delta$, since
 \begin{eqnarray*}
q(\lambda)&=&-\dfrac {\lambda^3}{a^2b^2}-{\dfrac {\left(a^2+2\,{b
}^{2}-{\delta}^2+{{\it x_c}}^2+{{\it y_c}}^2 \right) {\lambda}^2}{a^2b^2}}
\\
&&
-{\dfrac { \left(  \left( a^2+b^2 \right)  \left( b^2-\delta^2 \right) +
a^2{\it y_c}^2+b^2{\it x_c}^2+2\,{\it y_c}^2b^2
\right) \lambda}
{a^2b^2}}-{\dfrac {{{\it y_c}}^2 \left( a^2+b^2 \right)}{a^2}},
\end{eqnarray*}
the condition ${\bf Var}(q)=2$ is possible only if $b < \delta$. Thus, the conditions $\Delta=0$, ${\bf Var}(q)=2$ and $a''_0\neq 0$ are also sufficient.

 \item[9.]
By \ref{nueve} in Proposition  \ref{ch}, $\mathcal{NH}$ and $\mathcal{M}$  have four points of intersection iff either $f(\lambda)$ has three distinct negative roots which are not greater than $-a^2$, $ \lambda_1<\lambda_2<\lambda_3\leq-a^2$, with $a^2<b\delta$ or $f(\lambda)$ has two distinct positive roots not less than $b^2$, $ \lambda_1<0<b^2\leq\lambda_2<\lambda_3$ with $b<\delta$.

First, we prove that $f(\lambda) = 0$ has three distinct negative roots not greater than $-a^2$ with $a^2<b\delta$ iff $\Delta>0$ and ${\bf Var}(g)= 0$. It is clear that $f(\lambda) = 0$ has three distinct negative roots not greater than $-a^2$ iff $\Delta>0$ and ${\bf Var}(g)=0$. Next observe that this directly implies $a^2<b\delta$. If $a'_0\neq 0$, then all roots are less than $-a^2$ and hence $a^2<b\delta$ (see the point 5.). If $a'_0=0$, then a root is equal to $-a^2$ and the other two less than $-a^2$ and by solving the equation $g(\lambda)=0$, the fact that both should be less than $-a^2$ implies $b\delta>a^2$.

Next, let us see that $f(\lambda)$ has two distinct positive roots not less than $b^2$, $ \lambda_1<0<b^2\leq\lambda_2<\lambda_3$ with $b<\delta$ iff $\Delta>0$ and ${\bf Var}(q)> 0$.

Observe that if $a_0''\neq 0$, i.e. $b^2$ is not a root of $f(\lambda)$, then $f(\lambda)$ has two distinct positive roots greater than $b^2$ iff $\Delta>0$ and ${\bf Var}(q)= 2$. Since $a_0''<0$ and  $a_3''<0$, we have ${\bf Var}(q)\neq 1$ and so  the equality ${\bf Var}(q)= 2$ can be reduced to ${\bf Var}(q)>0$. If $a_0''= 0$, then $f(\lambda)$ has one positive root greater than $b^2$ iff $\Delta>0$ and ${\bf Var}(q)=1$. Since in this case ${\bf Var}(q)$ can not be equal to 2 because $q(\lambda)$ does always have a negative root, the equality ${\bf Var}(q)= 1$ can be also reduced to ${\bf Var}(q)>0$.

On the other hand, it is easy to see that if $b\geq \delta$, then $V(q)=0$.

 \item[10.]
Suppose first that $a_0'\neq 0$. In this case, by \ref{diez} in Proposition \ref{ch}, $\mathcal{M}$ is inside $\mathcal{NH}$ iff $f(\lambda) = 0$ has three distinct negative roots with $\lambda_1<-a^2<\lambda_2<\lambda_3<0$. This is equivalent to $\Delta>0$, ${\bf Var}(f)=0$ and ${\bf Var}(g)=2$.

Now let $a_0'= 0$, $a_1'\neq 0$ that is, $-a^2$ is a simple root of $f(\lambda)$. In this case, $\mathcal{M}$ is inside $\mathcal{NH}$ iff $f(\lambda) = 0$ has three distinct negative roots with $\lambda_1<-a^2=\lambda_2<\lambda_3<0$ or $\lambda_1=-a^2<\lambda_2<\lambda_3<0$. This is equivalent to $\Delta>0$, ${\bf Var}(f)=0$ and ${\bf Var}(g/\lambda)>0$.
Finally, let $a_0'= a_1'=0$. In this case, $-a^2$ is a double root of $f(\lambda)$ and
$$
g(\lambda)=-\dfrac{\lambda^3}{a^2b^2} +\dfrac{ (a^2 - b\delta)(a^2 + b\delta)\lambda^2}{b^2a^4}.
$$ Hence, in this last case, we can conclude that $\mathcal{M}$ is inside $\mathcal{NH}$ iff the three roots are $-a^2$, $-a^2$, $-\dfrac{b^2\delta^2}{a^2}$ with $a^2>b\delta$ iff $a_0'= a_1'=0$ and $a_2'>0$.

 \end{enumerate}
\end{proof}

This picture illustrates all relations described in Proposition \ref{ch} and \ref{ch2}.

 \begin{figure}[H]
\begin{center}
 \fbox{\includegraphics[scale=.13]{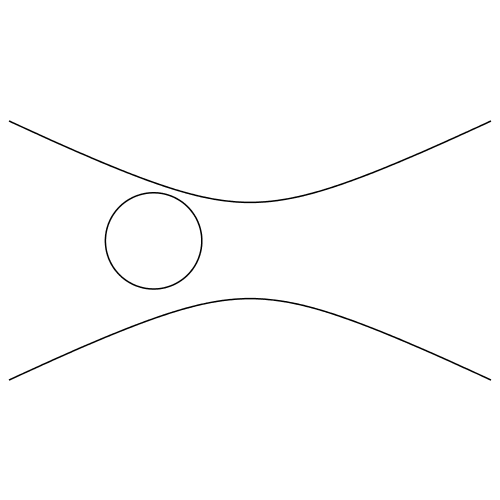} {\bf 1}} \hskip 3pt
  \fbox{\includegraphics[scale=.13]{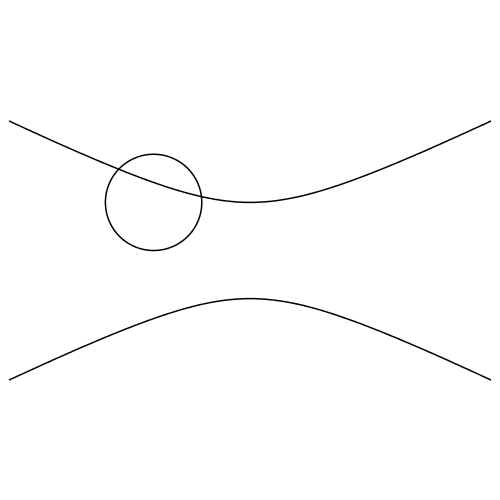} {\bf 2}} \hskip 3pt
 \fbox{\includegraphics[scale=.13]{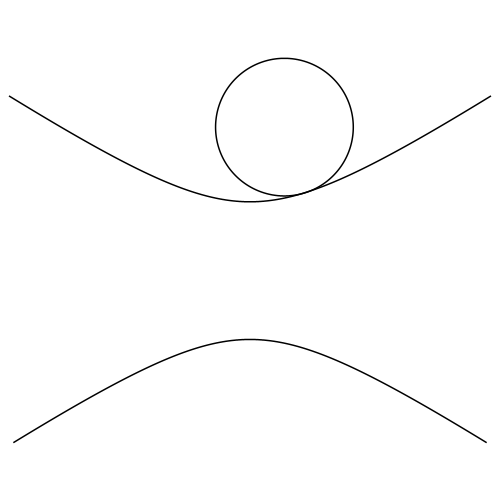} {\bf 3}} \hskip 3pt
  \fbox{\includegraphics[scale=.13]{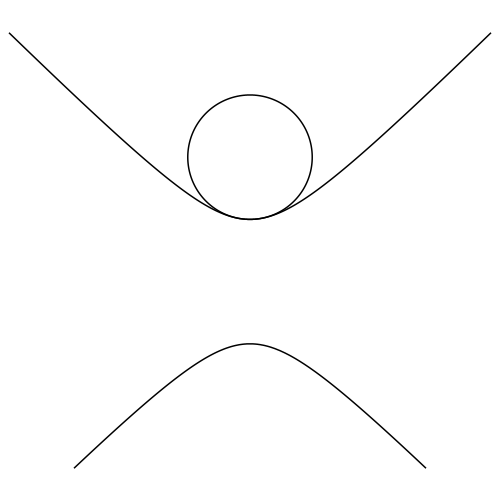} {\bf 3}} \hskip 3pt
 \fbox{\includegraphics[scale=.13]{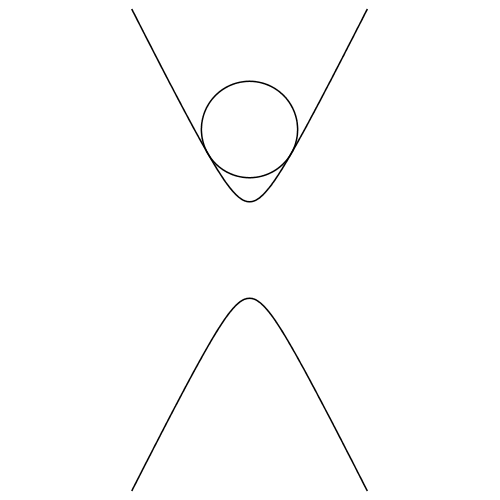} {\bf 4}} \\
 \vspace{2pt}
 \fbox{\includegraphics[scale=.13]{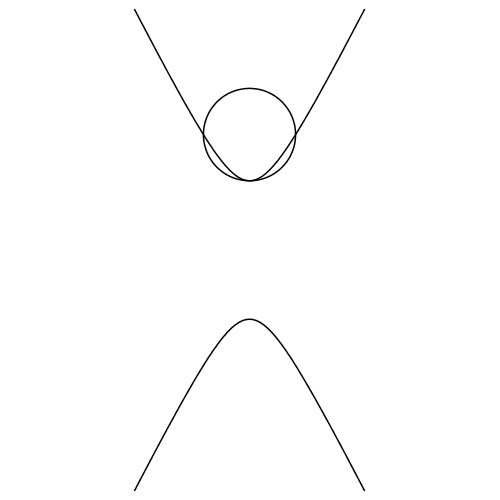} {\bf 5}} \hskip 3pt
 \fbox{\includegraphics[scale=.13]{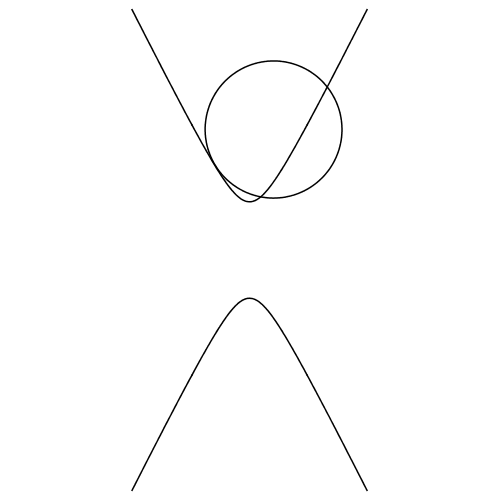} {\bf 5}} \hskip 3pt
 \fbox{\includegraphics[scale=.13]{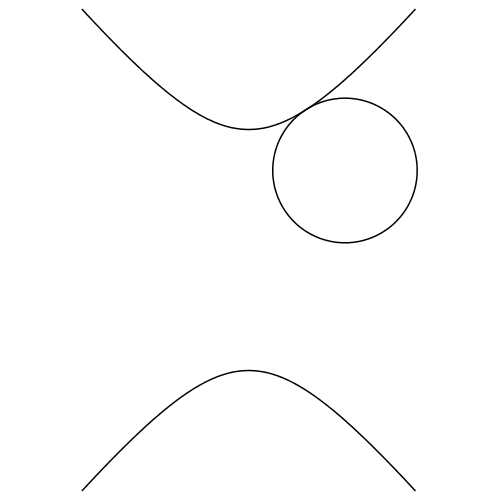} {\bf 6}} \hskip 3pt
 \fbox{\includegraphics[scale=.13]{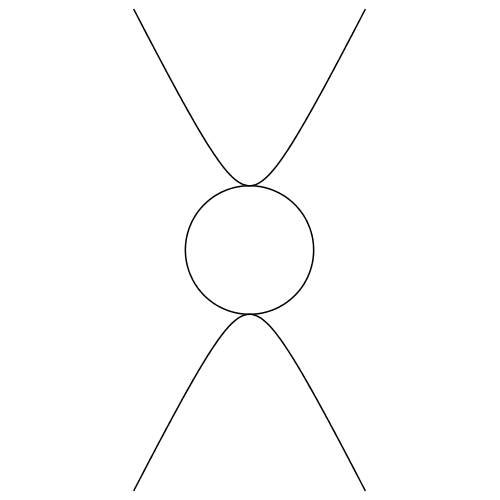} {\bf 7}} \\
 \vspace{2pt}
 \fbox{\includegraphics[scale=.13]{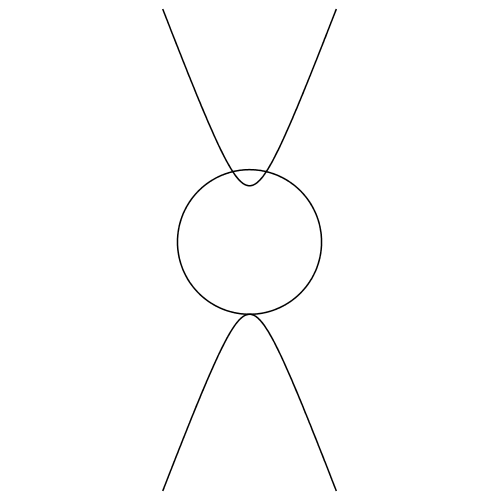} {\bf 8}} \hskip 2pt
 \fbox{\includegraphics[scale=.13]{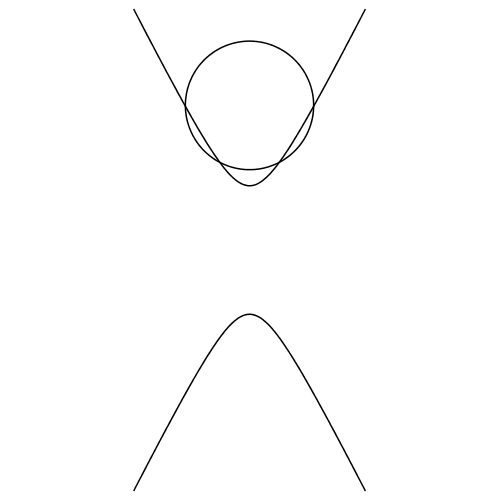} {\bf 9}} \hskip 2pt
 \fbox{\includegraphics[scale=.13]{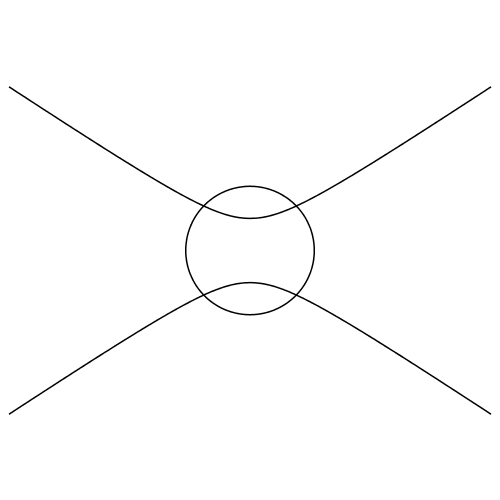} {\bf 9}} \hskip 2pt
 \fbox{\includegraphics[scale=.13]{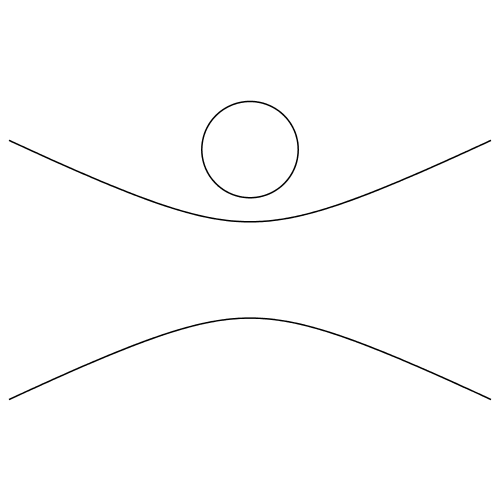} {\bf 10}} \hskip 2pt
  \fbox{\includegraphics[scale=.13]{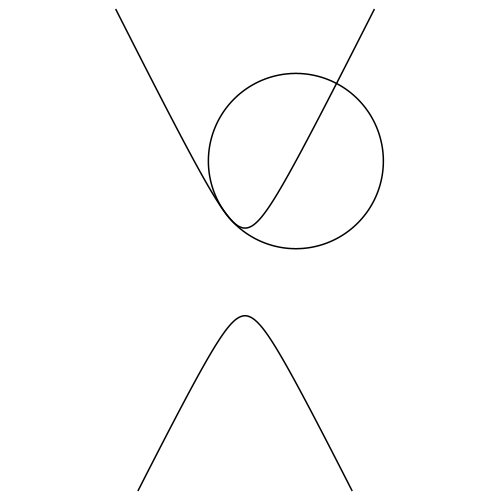} {\bf 11}}
\end{center}
\caption{ Circle and Hyperbola. Proposition \ref{ch}}
\end{figure}


\subsection{Relative position}\label{sec:nocanonicalEH}

Assume now that the ellipse $\mathcal{M}$ and the hyperbola $\mathcal{NH}$ are defined, respectively, by two symmetric matrices (not in canonical form),
$$
{\bf M}=\begin{pmatrix}
   a_{11} & a_{12} & a_{13} \\
   a_{12} & a_{22} & a_{23} \\
   a_{13} & a_{23} & a_{33}
\end{pmatrix},\quad
{\bf Nh}=\begin{pmatrix}
   b_{11} & b_{12} & b_{13} \\
   b_{12} & b_{22} & b_{23} \\
   b_{13} & b_{23} & b_{33}
\end{pmatrix}.
$$
We aim to characterize their relative position directly from their characteristic equation.
%
%
Following the definitions of Section \ref{cheq},
since $\mathcal{M}$ is an ellipse, we have that $T_1>0$; moreover, since it is real, we have $(a_{11}+a_{22})\,L_3<0.$
Assume that $a_{11}>0$, $a_{22} >0$, $L_3<0$ and the interior of the ellipse $\mathcal{M}$ is defined by $X^{T}MX<0$.
%
%
As for  $\mathcal{NH}$, we should have $T_2 <0$. Assume that $L_0<0$.

 By the same reasoning as in Section \ref{conics}, the given ellipse and hyperbola are transformed into the following circle and hyperbola,
$$
\mathcal{M}_2  : (x'-x'_c)+(y'-y'_c)^2 = \frac{-L_3}{T_1}\, ,
$$
$$
\mathcal{NH}_2  :H_0 x'^2 + H_2 y'^2  + H_5 =0 ,
$$
with $$
H_0=\dfrac{T+\sqrt{T^2-4T_1T_2}}{2T_1}\, ,\; H_2=\dfrac{T-\sqrt{T^2-4T_1T_2}}{2 T_1}\, ,\; H_5=\dfrac{L_0}{T_2}\, .
$$
Observe that, according to the properties of the hyperbola, we should have $H_0>0,H_2<0$ and $H_5>0.$

Then, since our goal is to apply the formulae from the previous section, dividing $\mathcal{NH}_2$ by $H_5$, we obtain the equation
$$
\dfrac{H_0}{H_5} x'^2 + \dfrac{H_2} {H_5}y'^2  + 1 =0,
$$
and following the notation of this section,
$$
\frac{1}{a^2}=\frac{H_0}{H_5},\quad -\frac{1}{b^2}=\frac{H_2}{H_5}.
$$
Hence, on the one hand, we have
$$
 (\det A)^{-2} \det(\lambda {\bf Nh}+ {\bf M})= (\det A)^{-2}(L_0 \lambda^3+ L_1 \lambda^2+  L_2 \lambda +L_3)= \det (\lambda {\bf Nh_1}+ {\bf M_1}),
$$
$$
\det (\lambda {\bf Nh_1}+ {\bf M_1})=\det (\lambda {\bf Nh_2}+ {\bf \hat{M}_2}),
$$
and on the other hand, since $\dfrac{\bf Nh_2}{H_5}$ and ${\bf \hat{M}_2}$ define a hyperbola and an ellipse like the ones introduced in this section,
$$
\det \Big(\lambda \frac{\bf Nh_2}{H_5}+ {\bf \hat{M}_2}\Big)=a_3 \lambda^3+ a_2 \lambda^2+  a_1 \lambda +a_0=(\det {\bf A})^{-2}\Big(\frac{L_0}{H_5^3} \lambda^3+ \frac{L_1}{H_5^2} \lambda^2+  \frac{L_2}{H_5} \lambda +L_3\Big).$$
Thus,
$$\Delta_f=(\det {\bf A})^{-8}\frac{\Delta_F}{H_5^6},$$
$$ a_3=(\det {\bf A})^{-2}\frac{L_0}{H_5^3},\; a_2=(\det {\bf A})^{-2}\frac{L_1}{H_5^2}, \; a_1=(\det {\bf A})^{-2}\frac{L_2}{H_5}, \; a_0=(\det {\bf A})^{-2} L_3.
$$
As a consequence, the coefficients of $g(\lambda)=f(\lambda-a^2)=a_3'\lambda^3+a_2'\lambda^2+a_1'\lambda+a_0'$ will be
\begin{align*}
a_2'&=(\det {\bf A})^{-2}\left(-3 L_0 \dfrac{a^2}{H_5^3 }+ \dfrac{L_1 }{H_5^2}\right),	\\
a_1'&=(\det {\bf A})^{-2}\left(3 L_0  \dfrac{a^4}{H_5^3 }- 2 L_1  \dfrac{a^2}{H_5^2} + \dfrac{L_2 }{H_5}\right), \\
a_0'&=(\det {\bf A})^{-2}\left(-L_0  \dfrac{a^6}{H_5^3} + L_1  \dfrac{a^4}{H_5^2} - L_2  \dfrac{a^2}{H_5} +  L_3\right) .
\end{align*}
Now, since $\dfrac{1}{a^2}=\dfrac{H_0}{H_5}$, we claim that

\medskip

$a_2'$ will have the same sign as $ -3L_0  + L_1 H_0=:J_5$,

\medskip
$a_1'$ will have the same sign as $3L_0 - 2L_1 H_0 + L_2H_0^2=:J_4$,

\medskip
$a_0'$  will have the same sign as $-L_0  + L_1 H_0 - L_2 H_0^2 + L_3 H_0^3=:J_3$.

\medskip
\noindent In a similar way, the coefficients of $q(\lambda)=f(\lambda+b^2)=a_3''\lambda^3+a_2''\lambda^2+a_1''\lambda+a_0''$ will be
\begin{align*}
a_2''&=(\det {\bf A})^{-2}\left(3 L_0 \dfrac{ b^2}{H_5^3} +  \dfrac{L_1 }{H_5^2}\right) , \\
a_1''&=(\det {\bf A})^{-2}\left(3 L_0  \dfrac{b^4}{H_5^3} + 2 L_1   \dfrac{b^2}{H_5^2} +  \dfrac{L_2 }{H_5}\right),\\
a_0'&'=(\det {\bf A})^{-2}\left(L_0  \dfrac{b^6}{H_5^3}+ L_1   \dfrac{b^4}{H_5^2} + L_2  \dfrac{ b^2}{H_5} +  L_3\right).
\end{align*}
Since $\dfrac{-1}{b^2}=\dfrac{H_2}{H_5}$, we claim that

\medskip

$a_2''$ will have the same sign as $ -3L_0/H_2  + L_1 =:K_5$,

\medskip
$a_1''$ will have the same sign as $ 3L_0-2H_2L_1+H_2^2L_2 =:K_4$,

\medskip
$a_0''$  will have the same sign as $L_0  - L_1 H_2 + L_2 H_2^2 - L_3 H_2^3=:K_3$.

\medskip
\noindent This therefore means that
$$
{\bf Var}(f)={\bf Var}(F), {\bf Var}(g)={\bf Var}(L_0,J_5,J_4,J_3),{\bf Var}(q)={\bf Var}(L_0,K_5,K_4,K_3) .
$$
On the other hand,
\begin{enumerate}
\item[i)] $-a^2$ is root of $f$ iff $0$ is a root of $g$ iff $J_3=0$.

\item[ii)] Since $a$, $b$ and $\delta$ are all strictly positive, $-a^2$ is triple root of $f$ iff $0$ is a triple root of $g$ iff $J_5=J_4=J_3=0$ iff $a^2=b\delta$ and $|y_c|=b+\delta$.

\item[iii)] With $a^2\neq b\delta$, $-a^2$ is root of $f$ and $-b\delta$ is a double root of $f$ iff $0$ is a root of $g$ and $a^2-b\delta$ is a double root of $g$ iff $J_3=0$, $a^2-b\delta=\frac{a_2'}{2\,a_3'} $ 
iff $J_3=0$,
$\sqrt{\dfrac{L_3 H_5}{H_2 T_1}} = H_5\dfrac{H_0L_1 - L_0}{2H_0L_0}$. This is equivalent to saying that $|y_c|=b+\delta$.
\end{enumerate}
Let $$J_1=\sqrt{\dfrac{L_3 H_5}{H_2 T_1}}  - H_5\dfrac{H_0L_1 - L_0}{2H_0L_0}.$$ Combining these facts, we conclude the following result.


\begin{theorem}\label{he}
The relationship between the ellipse $\mathcal{M}$ and the hyperbola $\mathcal{NH}$ are as follows:
\begin{enumerate}

\item $\mathcal{NH}$ and $\mathcal{M}$ are separated iff $\Delta>0,{\bf Var}(F)=2,{\bf Var}(L_0,K_5,K_4,K_3)=0$.

\item  $\mathcal{NH}$ and $\mathcal{M}$  have only two points of intersection iff $\Delta<0$.

\item  $\mathcal{NH}$ and $\mathcal{M}$  have only an inner tangent point iff 

\medskip

either   $\Delta=0$; ${\bf Var}(F)=0$, $J_3\neq  0 $, ${\bf Var}(L_0,J_5,J_4,J_3)=2$, 

or  $J_3=0$, $J_4= 0$, $J_5=0$,

or  $J_1=0$,  $J_3=0$, $J_4\neq 0$  and $J_5> 0$.

\item   $\mathcal{NH}$ and $\mathcal{M}$ have only two inner tangent points iff 
 $J_3=0 $, $J_4=0$ and $J_5<0$.

\item   $\mathcal{NH}$ and $\mathcal{M}$  have two points of intersection and an inner tangent point iff   \medskip

either $\Delta=0$, $\Delta'\neq 0$, ${\bf Var}(L_0,J_5,J_4,J_3)=0, J_3\neq 0$; 

or $\Delta=0$, ${\bf Var}(L_0,J_5,J_4)=0, J_3=0$, $J_5<0$ and $J_1=0$.

\item  $\mathcal{NH}$ and $\mathcal{M}$ have only one outer tangent point iff
$\Delta=0$, ${\bf Var}(F)=2$, ${\bf Var}(L_0,K_5,K_4,K_3)=0$ and $K_3\neq 0$.

\item $\mathcal{NH}$ and $\mathcal{M}$  have two outer tangent points iff $ K_3=0  ,  K_4=0 $.

\item   $\mathcal{NH}$ and $\mathcal{M}$ have two points of intersection and an outer tangent point
if $\Delta=0$, $  {\bf Var}(L_0,K_5,K_4,K_3)=2, K_3\neq 0$.

\item  $\mathcal{NH}$ and $\mathcal{M}$  have four points of intersection iff either $\Delta>0$ and ${\bf Var}(L_0,J_5,J_4,J_3)=0$ \medskip or $\Delta>0$ and ${\bf Var}(L_0,K_5,K_4,K_3)>0$.

\item  $\mathcal{M}$ is inside  $\mathcal{NH}$ iff \medskip either  $\Delta>0$, ${\bf Var}(F)= 0$, $J_3\neq 0$,  ${\bf Var}(L_0,J_5,J_4,J_3)=2$ or
 $\Delta>0$, ${\bf Var}(F)= 0$, $J_3= 0$, $ J_4\neq 0$, ${\bf Var}(L_0,J_5,J_4)>0$;
or $J_3=0 $, $J_4=0$, $J_5>0$.


\item  
$\mathcal{M}$ and $\mathcal{NH}$ have only an intersection point and an inner tangent point iff 
$\Delta=0$, $\Delta'=0$,  $J_3<0$ and $J_5<0$.
 
 \end{enumerate}

 \end{theorem}

Figure \ref{13} illustrates all relations described in Theorem \ref{he}.
  \begin{figure}[H]
\begin{center}
 \fbox{\includegraphics[scale=.16]{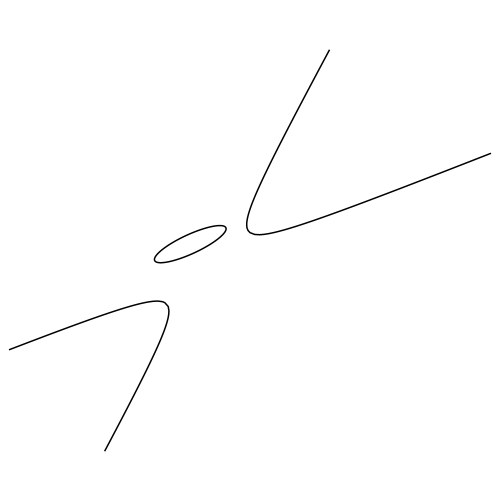} {\bf 1}} \hskip 3pt
  \fbox{\includegraphics[scale=.16]{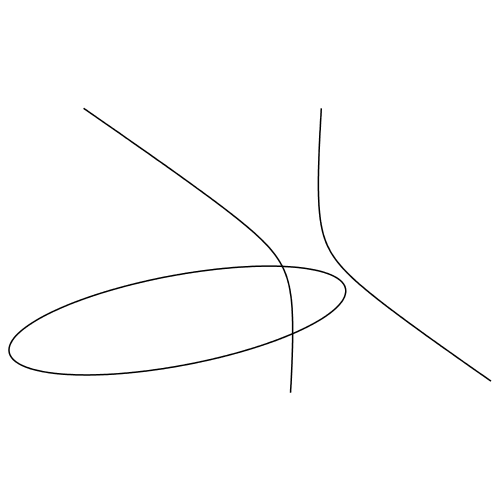} {\bf 2}} \hskip 3pt
 \fbox{\includegraphics[scale=.16]{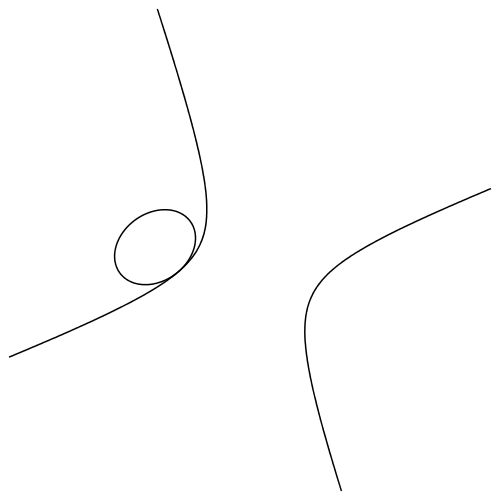} {\bf 3}} \hskip 3pt
  \fbox{\includegraphics[scale=.16]{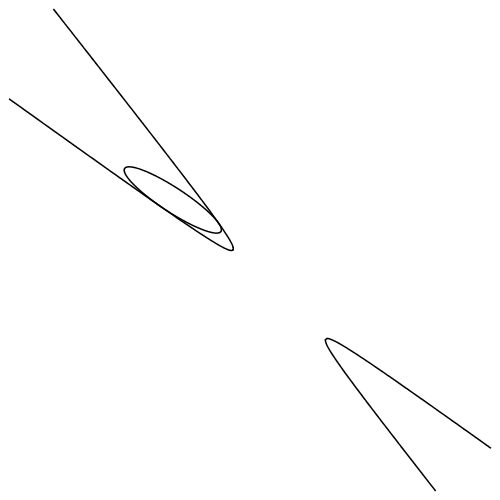} {\bf 4}} \\
 \vspace{2pt}
 \fbox{\includegraphics[scale=.16]{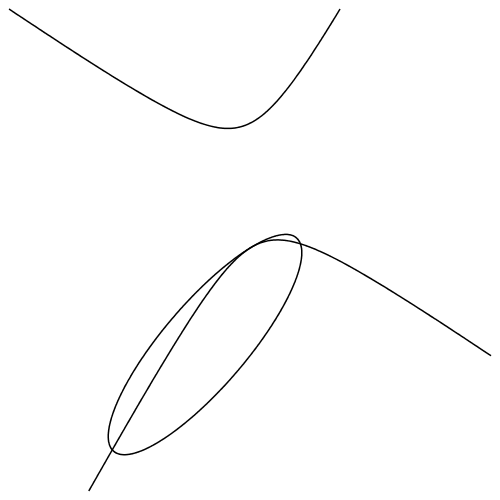} {\bf 5}} \hskip 3pt
 \fbox{\includegraphics[scale=.16]{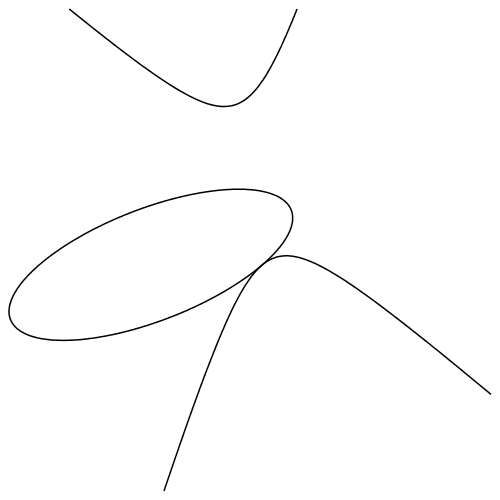} {\bf 6}} \hskip 3pt
  \fbox{\includegraphics[scale=.16]{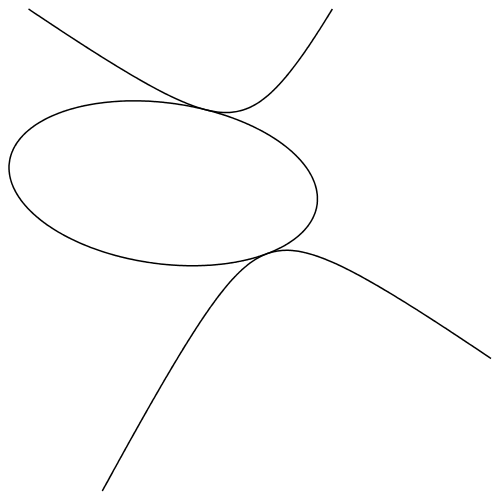} {\bf 7}} \\
 \vspace{2pt}
 \fbox{\includegraphics[scale=.16]{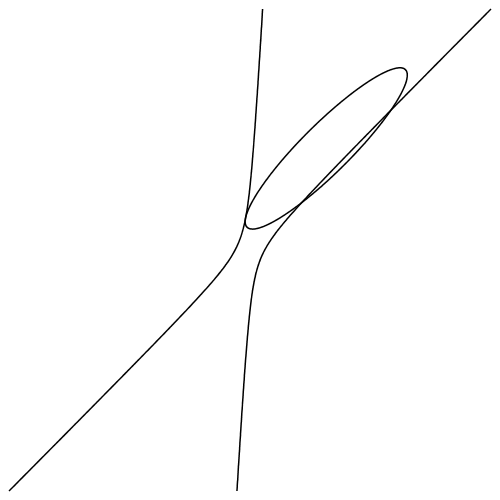} {\bf 8}}\hskip 3pt 
 \fbox{\includegraphics[scale=.16]{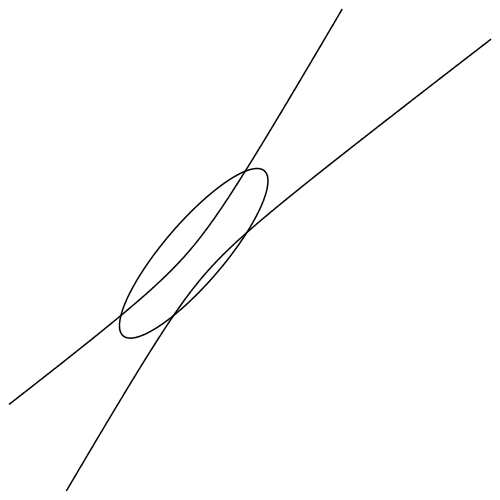} {\bf 9}} \hskip 3pt
 \fbox{\includegraphics[scale=.16]{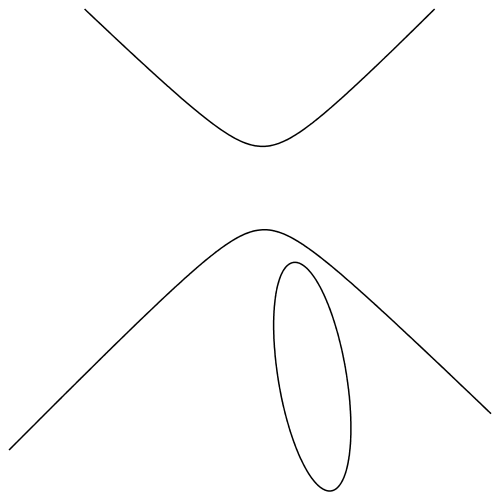} {\bf 10}}\hskip 3pt
 \fbox{\includegraphics[scale=.16]{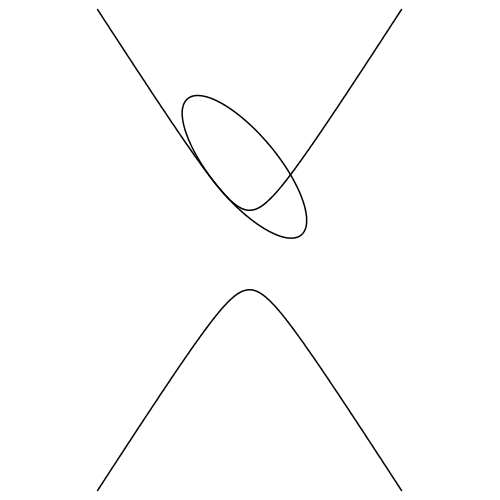} {\bf 11}}\hbox{}
\end{center}
\caption{Ellipse and Hyperbola - Theorem \ref{he}}
\label{13}
\end{figure}

\section{Conclusion}

We presented a method to determine the relative position of a pair of conics, for the two following cases:
 a parabola and an ellipse, and a hyperbola and an ellipse. The method requires the computation of certain coefficients which can be calculated directly from the coefficients of the implicit equations of both conics. Among those coefficients is, for example, the discriminant of the characteristic equation of the pair of conics. The sequence of signs of those coefficients characterizes the relative position of the conics. There are 9 relative positions for the pair parabola-ellipse, and 11 relative positions for the pair hyperbola-ellipse.
 
The method avoids the computation of the intersection points of the conics, and the computation of the distance between the conics. This approach allows working with conics whose implicit equation coefficients depend on a parameter. This may be very useful for application to situations with moving conics.




\end{document}